\documentclass[12pt]{article}
\usepackage{amsmath}      
\allowdisplaybreaks[3]  
\usepackage{graphicx} 
\usepackage{amsthm}
\usepackage{mathrsfs}
\usepackage{url}
\usepackage{bbm}
\usepackage[margin=20mm]{geometry}
\usepackage{graphicx}
\title{Non-relativistic limit of generalized relativistic Pauli operators by Feynman-Kac formulae}
\author{
Soichiro Sakamoto\\
Faculty of Mathematics\\
Kyushu University
}
\date{May 2026}

\theoremstyle{definition}
\newtheorem{definition}{Definition}[section]
\newtheorem{proposition}[definition]{Proposition}
\newtheorem{lemma}[definition]{Lemma}
\newtheorem{theorem}[definition]{Theorem}
\newtheorem{corollary}[definition]{Corollary}
\newtheorem{remark}[definition]{Remark}

\theoremstyle{plain}
\numberwithin{equation}{section}
\theoremstyle{remark}
\usepackage{amsmath, amssymb}
\begin{document}
\newcommand{\norm}[1]{\left\lVert #1 \right\rVert}
\maketitle

\begin{abstract}
The non-relativistic limit of a generalized relativistic Pauli operator
\[
H_c^{S,\alpha}
=
\left(
2c^{\beta}\bigl(\sigma\cdot(-i\nabla-a)\bigr)^2
+(mc^\gamma)^{2/\alpha}
\right)^{\alpha/2}
-mc^\gamma+V
\]
on $L^2(\mathbb{R}^3;\mathbb{C}^2)$ is investigated under the constraint
$2\alpha=\gamma\beta+\gamma^2$.
This operator generalizes the relativistic Pauli operator within the framework of Bernstein functions.
The associated heat semigroup $e^{-tH_c^{S,\alpha}}$ admits a Feynman--Kac representation involving Brownian motion, a subordinator, and a Poisson process.
Using this representation, we prove that the semigroup $e^{-tH_c^{S,\alpha}}$ converges strongly to $e^{-tH^{S,\alpha}}$ as $c\to\infty$, where the limiting generator is given by
\[
H^{S,\alpha}
=
\frac{\alpha}{2m^{\frac{2}{\alpha}-1}}
\bigl(\sigma\cdot(-i\nabla-a)\bigr)^2
+V.
\]
The non-relativistic limit of a generalized relativistic Schr\"odinger operator is also investigated.\end{abstract}

\section{Introduction}
In this paper, we study the relativistic Schr\"odinger operator with a vector potential $a$,
\[
H_c = \sqrt{ c^2 (-i\nabla - a)^2 + m^2 c^4 } - mc^2 + V
\]
acting on $L^2(\mathbb{R}^d)$, and the relativistic Pauli operator
\[
H_c^{\mathrm{S}} = \sqrt{ c^2 \bigl( \sigma \cdot (-i\nabla - a) \bigr)^2 + m^2 c^4 } - mc^2 + V
\]
acting on $L^2(\mathbb{R}^d;\mathbb{C}^2)$. Here $\sigma = (\sigma_1,\sigma_2,\sigma_3)$ denotes the Pauli matrices. The operator $H_c^{\mathrm{S}}$ incorporates the spin-$\tfrac12$ degree of freedom, whereas $H_c$ does not. Both $H_c$ and $H_c^{\mathrm{S}}$ describe relativistic particles minimally coupled to the vector potential $a$. The function $V$ represents an external potential, $c>0$ is the speed of light, and $m>0$ is the particle mass.

\;The purpose of this paper is to investigate the non-relativistic limit of the semi-groups $e^{-tH_c^\#}$ as $c\to\infty$, where the symbol $\#$ stands for either the scalar or spin case. Indeed, since
\[
\sqrt{c^2 u + m^2 c^4} - mc^2
= \frac{1}{2m} u + O\!\left(\frac{u^2}{m^3 c^2}\right),
\]
we formally expect that
\[e^{-tH_c^\#} \;\longrightarrow\; e^{-tH^\#}\]
in the strong sense as $c\to\infty$. Here
\begin{align*}
&H = \frac{1}{2m}(-i\nabla - a)^2 + V,\\
&H^{\mathrm{S}} = \frac{1}{2m}\bigl(\sigma \cdot (-i\nabla - a)\bigr)^2 + V,
\end{align*}
are the non-relativistic Schr\"odinger and Pauli operators, respectively.
In this paper we further introduce a class of generalized operators. Consider the function
\[
\Psi_c(u) = (2c^2 u + m^2 c^4)^{1/2} - mc^2,
\]
which is a Bernstein function, i.e, $(-1)^{n+1}\Psi_c^{(n)}(u)\ge0$ for all $u\ge0$. Also, $\Psi_c(0)=0$ is satisfied. Then the operators $H_c$ and $H_c^{\mathrm{S}}$ can be expressed as
\begin{align*}
&H_c = \Psi_c\!\left(\tfrac12 (-i\nabla - a)^2\right) + V,\\
&H_c^{\mathrm{S}} = \Psi_c\!\left(\tfrac12 (\sigma \cdot (-i\nabla - a))^2\right) + V.
\end{align*}
We generalize this construction by replacing $\Psi_c$ with a family of Bernstein functions
\[
\Psi_{\alpha,\beta,\gamma,c}(u)
= \bigl(2c^\beta u + (mc^\gamma)^{2/\alpha}\bigr)^{\alpha/2} - mc^\gamma,
\qquad 0<\alpha<2.
\]
Under the condition $2\alpha = \beta\gamma + \gamma^2$, one verifies that
\[
\Psi_{\alpha,\beta,\gamma,c}(u)
\longrightarrow
\frac{\alpha}{2m^{\frac{2}{\alpha}-1}}\,u
\]
as $c\to\infty$. The original function $\Psi_c$ corresponds to the special case $\alpha=1$ and $\beta=\gamma=2$.
Accordingly, we introduce the generalized operators:
\begin{align*}
&H_c^\alpha = \Psi_{\alpha,\beta,\gamma,c}\!\left(\tfrac12 (-i\nabla - a)^2\right) + V,\\
&H_c^{{\mathrm{S}},\alpha}
= \Psi_{\alpha,\beta,\gamma,c}\!\left(\tfrac12 (\sigma \cdot (-i\nabla - a))^2\right) + V.
\end{align*}
The main purpose of this paper is to prove that
\[
e^{-tH_c^{\#,\alpha}} \;\longrightarrow\; e^{-tH^{\#,\alpha}}
\]
as $c\to\infty$, where
\begin{align*}
&H^{\alpha} = \frac{\alpha}{2m^{\frac{2}{\alpha}-1}}(-i\nabla - a)^2 + V,\\
&H^{\mathrm{S},\alpha} = \frac{\alpha}{2m^{\frac{2}{\alpha}-1}}\bigl(\sigma \cdot (-i\nabla - a)\bigr)^2 + V.
\end{align*}

Our approach is based on the Feynman--Kac formula (FKF) for semi-groups $e^{-tH_c^{\#,\alpha}}$. The FKF for $e^{-tH_c}$ and $e^{-tH_c^{\mathrm{S}}}$ were derived in \cite{JHV2020}. The corresponding formulae for $e^{-tH_c^\alpha}$ and $e^{-tH_c^{{\mathrm{S}},\alpha}}$ can also be established by a slight modification of these arguments. In particular, the FKF for $e^{-tH_c^{{\mathrm{S}},\alpha}}$ involves three independent stochastic processes: a Brownian motion, a subordinator associated with the Bernstein function $\Psi_{\alpha,\beta,\gamma,c}$, and a spin process driven by a Poisson process.

Finally, we briefly review related results. 
A path integral representation of the semi-group 
$e^{-tH}$ is known as the Feynman--Kac--It\^o formula \cite{Sim1979}. 
Ichinose and Tamura extended this framework to the relativistic Hamiltonian $H_c$ in \cite{IT1986}, 
although their representation differs from the one employed in the present work. 
Based on these developments, Ichinose proved the non-relativistic limit in \cite{Ichi1987} in the spinless case. 
The path integral representations of the heat semi-groups for the Pauli operator and the relativistic Pauli operator are given by \cite{GGM1983,BJ1981} and \cite{GAM1991}, respectively, and it is generalized in \cite{HIL2009}.
The non-relativistic limit of the Dirac operator has also been investigated from an operator-theoretic perspective in \cite{Hun1975,Tha1992}.

\section*{Acknowledgments}
The author would like to express his deepest gratitude to Professor Fumio Hiroshima for his invaluable comments on the structure of this paper and for meticulously pointing out several mathematical errors in an earlier draft. His suggestions were instrumental in improving both the logical clarity and the mathematical rigor of this work.

\section{$\alpha/2$-relativistic subordinator}
\subsection{Bernstein functions and subordinators}
We introduce an $\alpha/2$-relativistic subordinator, which plays an important role in the subsequent stochastic arguments. Bernstein functions introduced below are treated in \cite{App2009, JHV2020}.
\begin{definition}\textbf{(Bernstein function)} Let \[\mathscr{B}=\biggl\{f\in C^\infty((0,\infty))\bigg|f(x)\ge0\;\text{and}\;(-1)^n\frac{d^n}{dx^n}f(x)\le0 \; \text{for}\;\text{all} \; n\in\mathbb{N}\biggr\}.\]\\
An element of $\mathscr{B}$ is called a Bernstein function. Furthermore, we define a subset of $\mathscr{B}$: 
\[\mathscr{B}_0=\{f\in\mathscr{B}\;|\lim_{x \to 0+}f(x)=0\}.\]
\end{definition}
We introduce a set of measures. Let $\mathscr{L}$ be the set of Borel measures $\lambda$ on $\mathbb{R}\backslash\{0\}$ such that
$\lambda((-\infty,0))=0$ and $\int_{\mathbb{R}\backslash\{0\}}(y\wedge1)\lambda(dy)<\infty.$
In particular, $\lambda\in\mathscr{L}$ is a Lévy measure, i.e, $\int_{\mathbb{R}\backslash\{0\}}(y\wedge1)^2\lambda(dy)<\infty.$
\begin{proposition} Let $\Psi\in\mathscr{B}_0$. Then there exists a unique $(b,\lambda)\in[0,\infty) \times\mathscr{L}$ such that 

\begin{equation}
\Psi(u)=bu\;+\;\displaystyle\int_0^\infty (1-e^{-uy})\lambda(dy).
\end{equation}
Conversely, the right-hand side of (2.1) is in $\mathscr{B}_0$ for each
$(b,\lambda)\in\mathbb{R}_+\times\mathscr{L}$.

\end{proposition}

There exists a fundamental relationship between subordinators and $\mathscr{B}_0$. 
This relationship is indispensable for establishing the path integral representation of generalized relativistic Pauli operators. Let $(X_t)_{t\ge0}$ be a random process on a probability space $(\Omega,\mathcal{F},P)$. From now on, we write 
\[\mathbb{E}_P^{x}[f(X_t)]=\int_{\Omega}f(X_t(\omega)+x)dP(\omega).\]
\begin{proposition}
Let $\Psi\in\mathscr{B}_0$, and $(b,\lambda)\in\mathbb{R}_+\times\mathscr{L}$ be 
a pair satisfying (2.1). Then there exists a unique subordinator $(T_t)_{t\ge0}$ on a probability space $(\Omega,\mathcal{F},P)$ such that 
\begin{equation}
\mathbb{E}_P^{0}[e^{-uT_t}]=e^{-t\Psi(u)}.
\end{equation}
Conversely, let $(T_t)_{t\ge0}$ be a subordinator on $(\Omega,\mathcal{F},P)$. Then there exists a unique $\Psi\in\mathscr{B}_0$ satisfying (2.2). 
\begin{proof}
See \cite[Proposition 3.99]{JHV2020}.    
\end{proof}
\end{proposition} 

\subsection{$\alpha/2$-relativistic subordinator}
Let $0<\alpha<2$, and choose $\beta,\gamma>0$ such that $2\alpha=\beta\gamma+{\gamma}^2$.
Let  
\begin{equation}
\Psi_c^{\alpha}(u)=(2c^{\beta}u+(mc^\gamma)^{2/\alpha})^{\alpha/2}-mc^\gamma, \quad u\ge0.    
\end{equation}
We can directly see that $\Psi_c^{\alpha}\in\mathscr{B}_0$ for each $c>0$. Let $(\Omega,\mathcal{F},P)$ be a probability space. We call $(T_t^c)_{t\ge0}$ satisfying (2.4) below an $\alpha/2$-relativistic subordinator parametrized by $c>0$:
\begin{equation}
\mathbb{E}_P^{0}[e^{-uT_t^c}]=e^{-t\Psi_c^{\alpha}(u)}.    
\end{equation}

The right-hand side of (2.4) can be represented by 
\begin{equation}
\exp(-t\Psi_c^{\alpha}(u))={\exp}\Biggl({\frac{{\alpha}(2c^{\beta})^{\frac{\alpha}{2}}t}{2\Gamma(1-{\alpha/2)}}\displaystyle\int_0^\infty(1-e^{-uy})e^{\frac{-(mc^\gamma)^{2/\alpha}}{2c^{\beta}}}\frac{dy}{y^{1+\alpha/2}}\Biggr)}.
\end{equation}
See \cite[Example 5.9]{Bog2009}. Since $T_t^c$ is a  Lévy process, we can see that Lévy triplet of $(T_t^c)_{t\ge0}$ is given by the form $(b,0,\mu)$.
Here $b{\in}\mathbb{R}$ and $\mu$ is a Lévy measure. By Lévy–Khintchine formula, we have 
\begin{equation}
\mathbb{E}_P^{0}[e^{iuT_t^c}]=\exp\Biggl(itbu+t\displaystyle\int_{\mathbb{R}\backslash\{0\}}(e^{iuy}-1-iuy\mathbbm{1}_{\{|y|\le1\}})\mu(dy)\Biggr).
\end{equation}
The function $\mathbb{R}{\ni}u\mapsto\mathbb{E}_P^{0}[e^{iuT_t^c}]$ can be analytically continued into the region $\{iu\in\mathbb{C}\;|\;u>0\}$. Hence we can obtain the Lévy triplet of $(T_t^c)_{t\ge0}$ by the analytical continuation of both sides of (2.6) and comparing the resulting expression with (2.5):
\begin{equation}
(b,0,\mu)=\left(\displaystyle\int_0^1y\nu_c(dy),0,\nu_c\right),
\end{equation}

where 
\[\nu_c(dy) = (2c^{\beta})^{\frac{\alpha}{2}}\frac{\alpha}{2\Gamma\left(1-\frac{\alpha}{2}\right)} \exp\left( -\frac{(mc^\gamma)^{2/\alpha}}{2c^\beta}y\right) \frac{1}{y^{1+\alpha/2}}dy.\] 
By (2.7) and Lévy-Itô decomposition, we can obtain a stochastic integral representation of $T_t^c$ by 
\begin{equation}
T_t^c=t\displaystyle\int_0^1z\nu_c(dz)+\displaystyle\int_1^{\infty}zN_c(t,dz)+\displaystyle\int_0^1z\tilde{N_c}(t,dz),
\end{equation}
where $N_c(t,dz)$ denote a Poisson random measure and $\tilde{N_c}(t,dz)$ a compensated Poisson random measure associated with $(T_t^c)_{t\ge0}$. Note that $(N_c(t,\cdot))_{t\ge0}$ is the Poisson process with intensity $\nu_c(A)=\mathbb{E}_P^{0}[N_c(1,A)]$ for
each $A\subset[0,\infty)$.
Furthermore, $\tilde{N_c}(t,A)=N_c(t,dz)-t\nu_c(A)$, and $(\tilde{N}_c(t,A))_{t\ge0}$ is a martingale for
each $A\subset[0,\infty)$. We investigate the exponent of $(T_t^c)_{t\ge0}$ and limiting behavior of its expectation as $c \to \infty$. This computation is used to justify the integrability of the Feynman–Kac formula of the relativistic Pauli operator, and to investigate the non-relativistic limit.

\begin{lemma}
Fix $c>0$. For $0<u<{\frac{(mc^{\gamma})^{\frac{2}{\alpha}}}{2c^\beta}}$ we have 
\begin{equation}
\mathbb{E}_P^0[e^{uT_t^c}]=\exp\Bigl(-t\bigl(\bigl(-2c^{\beta}u+(mc^{\gamma})^{\frac{2}{\alpha}}\bigr)^{\frac{\alpha}{2}}-mc^{\gamma}\bigr)\Bigr).
\end{equation}
In particular, for $u>0$, it follows that 
\begin{equation}
\lim_{c \to \infty}\mathbb{E}_P^0[e^{uT_t^c}]=\exp\left({t\frac{{\alpha}u}{m^{\frac{2}{\alpha}-1}}}\right)
\end{equation}
and $\displaystyle\sup_{c>0}\mathbb{E}_P^0[e^{uT_t^c}]<{\infty}$ for all $u\in\mathbb{R}$.
\begin{proof}
By the definition of $\nu_c$ we see that 
$\int_0^{\infty}(e^{uz}-1)\nu_c(dz)<\infty$ 
for $0<u<{\frac{(mc^{\gamma})^{\frac{2}{\alpha}}}{2c^\beta}}$, 
and 
\[
\displaystyle\int_0^{\infty}(e^{uz}-1)\nu_c(dz)=\bigl(-2c^{\beta}u+(mc^{\gamma})^{\frac{2}{\alpha}}\bigr)^{\frac{\alpha}{2}}-mc^{\gamma}
\]
follows. 
We write
\[
{Z_t^c}=\displaystyle\int_0^{t+}\displaystyle\int_1^{\infty}zN_c(ds,dz)+\displaystyle\int_0^{t+}\displaystyle\int_0^1z\tilde{N_c}(ds,dz),
\]
where $\int_0^tN_c(ds,dz)=N_c(t,dz)$. By Itô-formula for semimartingales, we have
\begin{align*}
e^{u{Z_t^c}}-1=&\displaystyle\int_0^{t+}\displaystyle\int_1^{\infty}e^{uZ_s^c}(e^{uz}-1)N_c(ds,dz)+\displaystyle\int_0^{t+}\displaystyle\int_0^1e^{uZ_s^c}(e^{uz}-1)\tilde{N_c}(ds,dz)\\
&+\displaystyle\int_0^t\displaystyle\int_0^1e^{uZ_s^c}(e^{uz}-uz-1)ds\nu_c(dz).
\end{align*}
Taking expectations of both sides above, we obtain that
\begin{align*}
\mathbb{E}_P^0[e^{u{Z_t^c}}]
&=
\;1+\mathbb{E}_P^0\left[\displaystyle\int_0^{t+}\displaystyle\int_1^{\infty}e^{uZ_s^c}(e^{uz}-1)N_c(ds,dz)\right]+\mathbb{E}_P^0\left[\displaystyle\int_0^t\displaystyle\int_0^1e^{uZ_s^c}(e^{uz}-uz-1)ds\nu_c(dz)\right]\\
&= 1+\mathbb{E}_P^0\left[\displaystyle\int_0^t\displaystyle\int_1^{\infty}e^{uZ_s^c}(e^{uz}-1)ds\nu_c(dz)\right]+\mathbb{E}_P^0\left[\displaystyle\int_0^t\displaystyle\int_0^1e^{uZ_s^c}(e^{uz}-uz-1)ds\nu_c(dz)\right]\\
&= 1+ \left(\displaystyle\int_0^{\infty}(e^{uz}-1)\nu_c(dz)-u\displaystyle\int_0^1z\nu_c(dz)\right)\displaystyle\int_0^t\mathbb{E}_P^0[e^{uZ_s^c}]ds.
\end{align*}
In the first equality, we used the martingale property of the compensated Poisson integrals.
The second equality follows from
\begin{align*}
\mathbb{E}_P^0\left[\int_0^{t+}\int_1^{\infty}e^{uZ_s^c}(e^{uz}-1)N_c(ds,dz)\right]
=\mathbb{E}_P^0\left[\int_0^t\int_1^{\infty}e^{uZ_s^c}(e^{uz}-1)ds\,\nu_c(dz)\right].
\end{align*}
Hence we can obtain that  
\begin{align*}
\mathbb{E}_P^0[e^{u{Z_t^c}}]=\exp\left(\displaystyle{t}\int_0^{\infty}(e^{uz}-1)\nu_c(dz)-\displaystyle{ut}\int_0^1z\nu_c(dz)\right).    
\end{align*}

By (2.8), it follows that
\begin{align*}
\mathbb{E}_P^0[e^{uT_t^c}]
&=
\exp\left(ut\displaystyle\int_0^1z\nu_c(dz)\right)\mathbb{E}_P^0[e^{u{Z_t^c}}]\\
&=\exp\left(ut\displaystyle\int_0^1z\nu_c(dz)\right)\exp\left(t\displaystyle\int_0^{\infty}(e^{uz}-1)\nu_c(dz)-ut\displaystyle\int_0^1z\nu_c(dz)\right)\\
&= \exp\left(t\displaystyle\int_0^{\infty}(e^{uz}-1)\nu_c(dz)\right)\\
&=\exp\Bigl(-t\bigl(\bigl(-2c^{\beta}u+(mc^{\gamma})^{\frac{2}{\alpha}}\bigr)^{\frac{\alpha}{2}}-mc^{\gamma}\bigr)\Bigr)
\end{align*}

from a similar computation to (2.5). Let $u>0$ be arbitrarily fixed. Since ${\frac{(mc^{\gamma})^{\frac{2}{\alpha}}}{2c^\beta}}\to{\infty}$ as $c\to{\infty}$, there exists $C>0$ such that ${\frac{(mc^{\gamma})^{\frac{2}{\alpha}}}{2c^\beta}}>u>0$ for all $c>C$. Then we have
\[
\lim_{c\to{\infty}}\mathbb{E}_P^0[e^{uT_t^c}]=\lim_{c\to{\infty}}\exp\Bigl(-t\bigl(\bigl(-2c^{\beta}u+(mc^{\gamma})^{\frac{2}{\alpha}}\bigr)^{\frac{\alpha}{2}}-mc^{\gamma}\bigr)\Bigr)={\exp\left({t\frac{{\alpha}u}{m^{\frac{2}{\alpha}-1}}}\right)}.
\]
Thus the lemma is proved.
\end{proof}
\end{lemma}

\begin{corollary}
We have
\begin{equation}
\lim_{c \to \infty}\mathbb{E}_P^0\left[{\left|T_t^c-\frac{{\alpha}t}{\;\;\;m^{\frac{2}{\alpha}-1}}\right|^n}\right]=0 \quad \text{for}\;\text{all} \;{n\in{\mathbb{N}}}.
\end{equation}
\begin{proof}
We denote $\frac{{\alpha}t}{\;\;\;m^{\frac{2}{\alpha}-1}}$ by $t_{\alpha}$. By Lemma.2.5, we have
\begin{align*}
\left(\mathbb{E}_P^0\left[\left|T_t^c - t_{\alpha}\right|^n\right]\right)^2
&\le 
2\left(\mathbb{E}_P^0\left[\left|T_t^c - t_{\alpha}\right|^n{\mathbbm{1}_{\{T_t^c\ge{t_{\alpha}}\}}}\right]\right)^2
+ 2\left(\mathbb{E}_P^0\left[\left|T_t^c - t_{\alpha}\right|^n{\mathbbm{1}_{\{T_t^c\le{t_{\alpha}}\}}}\right]\right)^2 \\
&\le 
2\mathbb{E}_P^0\left[(n!)^2\left(e^{(T_t^c - t_{\alpha})} - 1\right)^2\right]
+ 2\mathbb{E}_P^0\left[(n!)^2\left(e^{(t_{\alpha} - T_t^c)} - 1\right)^2\right]\to0 \quad (c\to{\infty}).
\end{align*}
Then the corollary is proved.
\end{proof}
\end{corollary}
\section{Non-relativistic limit of spinless case}
In this section, we consider the non-relativistic limit of generalized relativistic Schrödinger operators. First, we introduce the $\alpha/2$-relativistic Schrödinger operator with a vector potential via quadratic forms, and confirm its self-adjointness under suitable conditions.
\subsection{Generalized relativistic Schrödinger operators}
We define generalized relativistic Schrödinger operators under singular vector potential via the theory of quadratic forms. Let $\partial_\mu : \mathscr{S}'(\mathbb{R}^d) \to \mathscr{S}'(\mathbb{R}^d)$, 
$\mu = 1, \ldots, d$, be the distributional derivative with respect to the 
$\mu$-th coordinate on the tempered distribution space 
$\mathscr{S}'(\mathbb{R}^d)$. 
Let $p = -i\nabla$, where $\nabla = (\partial_1, \ldots, \partial_d)$. For a vector potential $a = (a_1, \ldots, a_d) : \mathbb{R}^d \to \mathbb{R}^d$, 
we formally define the Schrödinger operator with the vector potential $a$ by
\[
\frac{1}{2}(p - a)^2.
\]
We rigorously construct it as a self-adjoint operator on $L^2(\mathbb{R}^d)$ via a quadratic form. We introduce following assumptions on $a$:
\begin{align*}
\text{(A.1)} \quad & a \in (L^2_{\mathrm{loc}}(\mathbb{R}^d))^d, \quad  V\in{C_b(\mathbb{R}^d)};   \\
\text{(A.2)} \quad & a \in (L^2_{\mathrm{loc}}(\mathbb{R}^d))^d, 
\quad \nabla \cdot a \in L^1_{\mathrm{loc}}(\mathbb{R}^d), \quad  V\in{C_b(\mathbb{R}^d)}.
\end{align*}
By the next proposition, we can define Schrödinger operators under (A.1). 
\begin{proposition}
If $q$ is a closed semibounded quadratic form, there exists a unique self-adjoint operator $A$ such that 
\[
Q(q)=D(A), 
\quad q(\psi,\varphi)=(\psi,A{\varphi}),\quad \psi,\varphi \in D(A)
.\]
\end{proposition}
\quad Let $D_\mu=p_\mu-a_\mu, \mu=1,...,d$, where $p_\mu=-i\partial_\mu.$ Define the quadratic form $q$ by 
\begin{equation}
q(f,g)=\sum_{\mu=1}^d (D_{\mu}f,D_{\mu}g)
\end{equation}
and its quadratic domain is given by
\begin{equation}
Q(q)=\{f\in L^2(\mathbb{R}^d)\mid D_{\mu}f\in L^2(\mathbb{R}^d),\ \mu=1,\ldots,d\}.
\end{equation}
We can obtain a self-adjoint operator through the following lemma under the assumption (A.1). 
\begin{proposition}
Suppose (A.1). Then the quadratic form $q$ defined by (3.1) and (3.2) is a symmetric closed form. In particular, there exists a unique self-adjoint operator $h(a)$ satisfying
\[
D(h(a))=\{{f\in{L^2(\mathbb{R}^d)}}\;|\;q(f,\cdot)\in ({L^2(\mathbb{R}^d)})^*\}
,\]
\[
(h(a)f,g)=q(f,g)\;\;\; \text{for}\;f\in{D(h(a))},\;g\in{L^2(\mathbb{R}^d)}
.\]
\begin{proof}
See \cite[Lemma 1]{HC1981}.   
\end{proof}
\end{proposition}

\begin{definition}
Let
$\Psi_c^{\alpha}$ be the Bernstein function defined by (2.3).
Assume (A.1). We define
\[
D(H_c^{\alpha})=D(\Psi_c^{\alpha}(h(a))),\quad \quad
H_c^\alpha = \Psi_c^{\alpha}(h(a)) + V,
\]
where $h(a)$ is defined in Proposition 3.2. We call $H_c^\alpha$ the $\alpha/2$-relativistic Schrödinger operator with vector potential $a$.
\end{definition}
Since $V$ is bounded, $H_c^\alpha$ is a self adjoint operator on $D(\Psi_c^{\alpha}(h(a)))$.
\subsection{Feynman-Kac-Itô formula}
Let $(B_t)_{t\ge0}$ be the $d$-dimensional Brownian motion starting at $x\in{\mathbb{R}^d}$ on the Wiener space $(\mathscr{X},\mathcal{B}(\mathscr{X}),\mathcal{W}^x)$, where $\mathscr{X}=C([0,\infty);\mathbb{R}^d)$ and $\mathcal{W}^x$ is the Wiener measure. To derive the Feynman-Kac-Itô formula, we need to check that a stochastic integral of $a$ can be defined under the assumption (A.2).

\begin{lemma}
Suppose (A.2). Then \[\mathcal{W}^x\left(\Biggl|\displaystyle\int_0^ta(B_s)\cdot{dB_s}+\frac{1}{2}\displaystyle\int_0^t{\nabla}{\cdot}a(B_s)ds\Biggr|<\infty\right)=1\]
holds for all $t\ge0$.
\begin{proof}
See \cite[Lemma 3.6]{HIL2009}.   
\end{proof}
\end{lemma}
We write
\[
\int_0^t a(B_s)\cdot dB_s
+ \frac{1}{2}\int_0^t \nabla \cdot a(B_s)\,ds
=
\int_0^t a(B_s)\circ dB_s.
\]

\quad Let $(T_t^c)_{t\ge0}$ be the $\alpha/2$-relativistic subordinator defined in (2.4) on a probability space $(\Omega,\mathcal{F},P)$.
From the above preparation, we can obtain the following Feynman--Kac--Itô formula for $e^{-tH_c^\alpha}$.

\begin{proposition}
Suppose (A.2). Then we have \[(f,e^{-tH_{c}^\alpha}g)=\displaystyle\int_{\mathbb{R}^d}\mathbb{E}_{\mathcal{W},P}^{x,0}\Bigl[\overline{f(B_0)}g(B_{T_t^c})e^{-i{\int_0^{T_t^c}a(B_s)\circ{dB_s}}}e^{-\int_0^{t}
V(B_{T_s^c}){ds}}\Bigr]dx.\]
\begin{proof}
See \cite[Theorem~4.208]{JHV2020}.     
\end{proof}
\end{proposition}
From Proposition 3.5, it follows that $\|e^{-tH_{c}^\alpha}\|$ is uniformly bounded for all $c>0$ and for all $t\ge0$. 
\subsection{Non-relativistic limit}
With the preparation of previous sections, we give a stochastic proof of the non-relativistic limit of 
$e^{-t{H_{c}^\alpha}}$ by using the Feynman–Kac–Itô formula. Here we set 
\[H_{\alpha}=\frac{{\alpha}}{m^{\frac{2}{\alpha}-1}}h(a)+V.\]
\begin{lemma}
Suppose $a\in{(L^2(\mathbb{R}^d))^d}$, ${\nabla}\cdot{a}\in L^1(\mathbb{R}^d)$ and $V\in{C_b(\mathbb{R}^d)}$. Then we have
\begin{align*}
s-\lim_{c\to\infty}e^{-tH_{c}^\alpha}=e^{-tH_{\alpha}}\;\;\;\text{for}\;\text{all}\;t\ge0.
\end{align*}
\allowdisplaybreaks
\begin{proof} Suppose $f,g\in{C^{\infty}_0(\mathbb{R}^d)}$. We have
\begin{align*}
&\bigl|(f,e^{-tH_{c}^\alpha}g)
-(f,e^{-tH_{\alpha}}g)\bigr| \\
&\le
2e^{t\delta}\int_{\mathbb{R}^d}
\mathbb{E}_{\mathcal{W},P}^{x,0}
\Bigl[
|f(B_0)|
|g(B_{T_t^c})-g(B_{t_{\alpha}})|
\Bigr] dx
\\
&\quad+
e^{t\delta}\int_{\mathbb{R}^d}
\mathbb{E}_{\mathcal{W},P}^{x,0}
\Bigl[
|f(B_0)|
|g(B_{t_{\alpha}})|
\bigl|
e^{-i\int_0^{T_t^c}a(B_s)\circ dB_s}
-
e^{-i\int_0^{t_{\alpha}}a(B_s)\circ dB_s}
\bigr|
\Bigr] dx
\\
&\quad\quad+
2\int_{\mathbb{R}^d}
\mathbb{E}_{\mathcal{W},P}^{x,0}
\Bigl[
|f(B_0)|
|g(B_{t_{\alpha}})|
\bigl|
e^{-\int_0^{t}V(B_{T_s^c})\,ds}
-
e^{-\int_0^{t}V(B_{s_{\raisebox{-0.2ex}{\scalebox{0.5}{$\alpha$}}}})\,ds}
\bigr|
\Bigr] dx.
\end{align*}

In this proof, we set $\delta={\lVert V \rVert}_{\infty}$. Then the convergence of the first and the third term follows from \cite[Proposition 4.230]{JHV2020}. We will verify the convergence of the second one:

\begin{align*}
&e^{t\delta}\int_{\mathbb{R}^d}
\mathbb{E}_{\mathcal{W},P}^{x,0}
\Bigl[
|f(B_0)|
|g(B_{t_{\alpha}})|
\bigl|
e^{-i\int_0^{T_t^c}a(B_s)\circ dB_s}
-
e^{-i\int_0^{t_{\alpha}}a(B_s)\circ dB_s}
\bigr|
\Bigr] dx\\
&\le e^{t\delta}\displaystyle\int_{\mathbb{R}^d}|f(x)|
 \mathbb{E}_{\mathcal{W},P}^{x,0}
 \left[ |g(B_{t_{\alpha}})|
 \left|\int_{t_{\alpha}}^{T_t^c} a(B_s)\cdot dB_s \right| \right] dx \\
&\qquad\qquad\qquad\qquad
 +\frac{e^{t\delta}}{2}\displaystyle\int_{\mathbb{R}^d}|f(x)|
 \mathbb{E}_{\mathcal{W},P}^{x,0}
 \left[ |g(B_{t_{\alpha}})|
 \left|\int_{t_{\alpha}}^{T_t^c} \nabla\!\cdot a(B_s)\, ds \right| \right] dx.
\end{align*}

We show the convergence of the first term. The second one follows similarly to the first one. We have
\begin{align*}
&e^{t\delta}\displaystyle\int_{\mathbb{R}^d}|f(x)|
 \mathbb{E}_{\mathcal{W},P}^{x,0}
 \left[ |g(B_{t_{\alpha}})|
 \left|\int_{t_{\alpha}}^{T_t^c} a(B_s)\cdot dB_s \right| \right] dx\\
& = e^{t\delta}\displaystyle\int_{\mathbb{R}^d}|f(x)|
 \mathbb{E}_{\mathcal{W},P}^{x,0}
 \left[ |g(B_{t_{\alpha}})|
 \left|\int_{t_{\alpha}}^{T_t^c} a(B_s)\cdot dB_s \right|\mathbbm{1}_{\{T_t^c\ge{t_{\alpha}}\}}\right] dx \\
&\qquad\qquad
 + e^{t\delta}\displaystyle\int_{\mathbb{R}^d}|f(x)|
 \mathbb{E}_{\mathcal{W},P}^{x,0}
 \left[ |g(B_{t_{\alpha}})|
 \left|\int_{t_{\alpha}}^{T_t^c} a(B_s)\cdot dB_s \right|\mathbbm{1}_{\{T_t^c\le{t_{\alpha}}\}} \right] dx. 
\end{align*}

It suffices to show the convergence of the first term, since the second term can be treated in the same way as the first. By Corollary 2.5, we see that
\begin{align*}
&\left(e^{t\delta}\displaystyle\int_{\mathbb{R}^d}|f(x)|
 \mathbb{E}_{\mathcal{W},P}^{x,0}
 \left[ |g(B_{t_{\alpha}})|
 \left|\int_{t_{\alpha}}^{T_t^c} a(B_s)\cdot dB_s \right|\mathbbm{1}_{\{T_t^c\ge{t_{\alpha}}\}}\right] dx\right)^2 \\
& \le e^{2t\delta}{\norm{f}}_{L^2}^2{\norm{g}}_{L^\infty}^2\displaystyle\int_{\mathbb{R}^d}
 \mathbb{E}_{\mathcal{W},P}^{x,0}
 \left[ \left|\int_{t_{\alpha}}^{T_t^c} a(B_s)\cdot dB_s \right|^2\mathbbm{1}_{\{T_t^c\ge{t_{\alpha}}\}}\right] dx \\
& = e^{2t\delta}{\norm{f}}_{L^2}^2{\norm{g}}_{L^\infty}^2\displaystyle\int_{\mathbb{R}^d}
 \mathbb{E}_{\mathcal{W},P}^{x,0}
 \left[ \int_{t_{\alpha}}^{T_t^c} |a(B_s)|^2ds \mathbbm{1}_{\{T_t^c\ge{t_{\alpha}}\}}\right] dx \\
& \le e^{2t\delta}{\norm{f}}_{L^2}^2{\norm{g}}_{L^\infty}^2{\norm{a}}_{L^2}^2\mathbb{E}_P^{0}[|T_t^c-t_{\alpha}|^2] \to\;\; 0 \qquad (c\to\infty).
\end{align*}
In the second equality, we used the Itô-isometry. Hence $e^{-tH_{c}^\alpha}{\to}e^{-tH_{\alpha}}$ on $C_0^{\infty}(\mathbb{R}^d)$ as $c\to{\infty}$. By the limiting argument by using the uniform boundness of $\|e^{-tH_c^{\alpha}}\|$, the lemma is proved.
\end{proof}
\end{lemma}

To extend this to the vector potential $a$ satisfying (A.2), we introduce the following lemma. 

\begin{lemma}
Let $\varepsilon>0$ be arbitrary. Suppose (A.2). Then there exists $(a_R)_{R>0}\subset{L^2(\mathbb{R}^d)}$ such that 
\begin{equation}
\hspace*{-1em}\mathcal{W}^x\left(\left|\displaystyle\int_0^ta(B_s)\circ{dB_s}-\displaystyle\int_0^ta_R(B_s)\circ{dB_s}\right|\ge\varepsilon\right)\le4\left(\frac{2}{\sqrt{2{\pi}t}}\displaystyle\int_R^{\infty}e^{-\frac{y^2}{2t}}dy\right)^d
\end{equation}
for $x\in\mathbb{R}^d$ and $t\ge0$. In particular, $\displaystyle\int_0^ta_R(B_s)\circ{dB_s}\to\displaystyle\int_0^ta(B_s)\circ{dB_s}$ in probability as $R\to{\infty}$.

\begin{proof}
The proof of the lemma is based on \cite[Proposition 3.7]{HIL2009}. Let $\chi \in C_0^\infty(\mathbb{R}^d)$ be $0 \le \chi \le 1$, $\chi(x)=1$ for $|x|<1$, and $\chi(x)=0$ for $|x|\ge 2$. Let $\chi_R=\chi(\frac{x_1}{R})\cdots\chi(\frac{x_d}{R})$, $R>0$.
Define $a_R$ by $a_R(x)=\chi_R(x)a(x)$, $x\in\mathbb{R}^d$. Furthermore, let
\[\Omega_+(R)=\{\omega\in\mathscr{X}\displaystyle|\;\max\limits_{0\le s\le t,\; 1\le \mu\le d} B_s^\mu\le{R}\}
,\]
\[
\;\;\Omega_-(R)=\{\omega\in\mathscr{X}\displaystyle|\;\min\limits_{0\le s\le t,\; 1\le \mu\le d} B_s^\mu\ge{-R}\}
\]

and
\begin{align*}
I(R) &= \left|\displaystyle\int_0^t a(B_s) \cdot dB_s - \displaystyle\int_0^t a_R(B_s) \cdot dB_s \right|, \\
J(R) &= \left|\frac{1}{2}\displaystyle\int_0^t {\nabla}{\cdot}a(B_s) ds - \frac{1}{2}\displaystyle\int_0^t {\nabla}{\cdot}a(B_s) ds \right|.
\end{align*}

Then we see that
\begin{align*}
&\mathcal{W}^x\left(\left|\displaystyle\int_0^ta(B_s)\circ{dB_s}-\displaystyle\int_0^ta_R(B_s)\circ{dB_s}\right|\ge\varepsilon\right)\le{\mathcal{W}^x\biggl(I(R)\ge{\frac{\varepsilon}{2}}\biggr)}+{\mathcal{W}^x\biggl(J(R)\ge{\frac{\varepsilon}{2}}\biggr)}.
\end{align*}

To prove (3.3), we only show that the first term satisfies 
\[\mathcal{W}^x\biggl(I(R)\ge{\frac{\varepsilon}{2}}\biggr)\le2\left(\frac{2}{\sqrt{2{\pi}t}}\displaystyle\int_R^{\infty}e^{-\frac{y^2}{2t}}dy\right)^d. 
\]
The second term is dominated by the same bound and can be estimated in the same way as the first.
We can see that $\max\limits_{0\le s\le t,\; 1\le \mu\le d} B_s^\mu$ and $\min\limits_{0\le s\le t,\; 1\le \mu\le d} B_s^\mu$ have the same distribution. Hence we have 
\[
\mathcal{W}^x(\Omega_-(R))=\mathcal{W}^x(\Omega_+(R))= \prod_{\mu=1}^d\mathcal{W}^x(|B_t^\mu|\le{R})=\left(\frac{2}{\sqrt{2{\pi}t}}\displaystyle\int_R^{\infty}e^{-\frac{y^2}{2t}}dy\right)^d.
\]

It follows that $I(R)=0$ on $\Omega_+(R)\cap\Omega_-(R)$ for all $0\le{s}\le{t}$. Thus we obtain that
\begin{align*}
\mathcal{W}^x\biggl(I(R)\ge{\frac{\varepsilon}{2}}\biggr)&=\mathcal{W}^x\biggl(I(R)\ge{\frac{\varepsilon}{2}},\Omega_+(R)^c\cup\;\Omega_-(R)^c\biggr)\\
&\le2\left(\frac{2}{\sqrt{2{\pi}t}}\displaystyle\int_R^{\infty}e^{-\frac{y^2}{2t}}dy\right)^d .
\end{align*}
Thus the lemma is proved.
\end{proof}
\end{lemma}

\begin{theorem}
Suppose (A.2). Then we have 

\[s-\lim_{c\to\infty}e^{-tH_{c}^\alpha}=e^{-tH_{\alpha}}\;\;\;\text{for}\;\text{all}\;t\ge0.\]

\begin{proof}
Following the limiting argument in the proof of Lemma 3.6, it suffices to show that for $f,g\in{C^{\infty}_0(\mathbb{R}^d)}$, 
\[
\displaystyle\int_{\mathbb{R}^d}\mathbb{E}_{\mathcal{W},P}^{x,0}\Bigl[|f(B_0)|g(B_{t_{\alpha}})||e^{-i{\int_0^{T_t^c}a(B_s)\circ{dB_s}}}-e^{-i{\int_0^{t_{\alpha}}a(B_s)\circ{dB_s}}}|\Bigr]dx \to 0 
\]
as $c\to{\infty}$. We denote $\int_0^{t} a(B_s)\circ dB_s$ and $\int_0^{t} a_R(B_s)\circ dB_s$ by $A_t$ and $A_t^R$, respectively. We see that
\begin{equation}
\hspace{-1em}
\begin{aligned}
&\int_{\mathbb{R}^d}\mathbb{E}_{\mathcal{W},P}^{x,0}\!\Bigl[
|f(B_0)||g(B_{t_{\alpha}})|
\bigl|e^{-i\int_0^{T_t^c} a(B_s)\circ dB_s}
      -e^{-i\int_0^{t_{\alpha}} a(B_s)\circ dB_s}\bigr|
\Bigr]dx \\[0.5em]
&\qquad\le
\int_{\mathbb{R}^d}\mathbb{E}_{\mathcal{W},P}^{x,0}\!\Bigl[
|f(B_0)||g(B_{t_{\alpha}})|
\bigl|e^{-iA_{T_t^c}}-e^{-iA_{T_t^c}^R}\bigr|
\Bigr]dx \\[0.5em]
&\qquad\qquad+
\int_{\mathbb{R}^d}\mathbb{E}_{\mathcal{W},P}^{x,0}\!\Bigl[
|f(B_0)||g(B_{t_{\alpha}})|
\bigl|e^{-i{A_{T_t^c}^R}}-e^{-iA_{t_{\alpha}}^R}\bigr|
\Bigr]dx \\[0.5em]
&\qquad\qquad\qquad+
\int_{\mathbb{R}^d}\mathbb{E}_{\mathcal{W},P}^{x,0}\!\Bigl[
|f(B_0)||g(B_{t_{\alpha}})|
\bigl|e^{-iA_{t_{\alpha}}^R}-e^{-iA_{t_{\alpha}}}\bigr|
\Bigr]dx.
\end{aligned}
\refstepcounter{equation}
\tag*{\hspace{2em}\raisebox{-4.0\baselineskip}{(\theequation)}}
\end{equation}
A positive constant $R$ will be chosen later. 
First, we evaluate the first term of the right-hand side of (3.4). Let $0<\varepsilon<{t_{\alpha}}$. 
\allowdisplaybreaks
Then we have
\begin{align}
\hspace{-1em}
&\left(\int_{\mathbb{R}^d}\mathbb{E}_{\mathcal{W},P}^{x,0}\!\Bigl[
|f(B_0)||g(B_{t_{\alpha}})|
\bigl|e^{-iA_{T_t^c}}-e^{-iA_{T_t^c}^R}\bigr|
\Bigr]dx\right)^2 \nonumber\\
&\le {\norm{f}}^2_{L^2}
\int_{\mathbb{R}^d}\left(\mathbb{E}_{\mathcal{W},P}^{x,0}\!\Bigl[
|g(B_{t_{\alpha}})|
\bigl|e^{-iA_{T_t^c}}-e^{-iA_{T_t^c}^R}\bigr|
\Bigr]\right)^2dx \nonumber\\
&\le {\norm{f}}^2_{L^2}
\int_{\mathbb{R}^d}\mathbb{E}_{\mathcal{W}}^{x}\!\Bigl[
|g(B_{t_{\alpha}})|^2\Bigr]
\mathbb{E}_{\mathcal{W},P}^{x,0}\!\Bigl[
\bigl|e^{-iA_{T_t^c}}-e^{-iA_{T_t^c}^R}\bigr|^2
\Bigr]dx \nonumber\\
&= {\norm{f}}^2_{L^2}\mathbb{E}_P^{0}\Bigl[
\mathbbm{1}_{\{|T_t^c-t_{\alpha}|\ge\varepsilon\}}
\int_{\mathbb{R}^d}\mathbb{E}_{\mathcal{W}}^{x}\!\Bigl[
|g(B_{t_{\alpha}})|^2\Bigr]
\mathbb{E}_{\mathcal{W}}^{x}\!\Bigl[
\bigl|e^{-iA_{T_t^c}}-e^{-iA_{T_t^c}^R}\bigr|^2
\Bigr]dx\Bigr] \nonumber\\
&\,\,\,\,\,\,\,
+{\norm{f}}^2_{L^2}\mathbb{E}_P^{0}\Bigl[
\mathbbm{1}_{\{|T_t^c-t_{\alpha}|\le\varepsilon\}}
\int_{\mathbb{R}^d}\mathbb{E}_{\mathcal{W}}^{x}\!\Bigl[
|g(B_{t_{\alpha}})|^2\Bigr]
\mathbb{E}_{\mathcal{W}}^{x}\!\Bigl[
\bigl|e^{-iA_{T_t^c}}-e^{-iA_{T_t^c}^R}\bigr|^2
\Bigr]dx\Bigr].
\end{align}
We can directly verify that the first term of the most right-hand side of (3.5) is bounded as follows:
\begin{align}
&{\norm{f}}^2_{L^2}\mathbb{E}_P^{0}\Bigl[
\mathbbm{1}_{\{|T_t^c-t_{\alpha}|\ge\varepsilon\}}
\int_{\mathbb{R}^d}\mathbb{E}_{\mathcal{W}}^{x}\!\Bigl[
|g(B_{t_{\alpha}})|^2\Bigr]
\mathbb{E}_{\mathcal{W}}^{x}\!\Bigl[
\bigl|e^{-iA_{T_t^c}}-e^{-iA_{T_t^c}^R}\bigr|^2
\Bigr]dx
\Bigr] \nonumber
\le
4{\norm{f}}^2_{L^2}{\norm{g}}^2_{L^2}
\mathbb{E}_P^{0}\Bigl[
\mathbbm{1}_{\{|T_t^c-t_{\alpha}|\ge\varepsilon\}}
\Bigr].
\end{align}

Next, we evaluate the second term of the most right-hand side of (3.5). Let $\eta > 0$ be arbitrary. We set 
\begin{align*}
\chi_{R,c}^1=\mathbbm{1}_{\{|A_{T_t^c}-A_{T_t^c}^R|\le \eta\}},
\quad \chi_{R,c}^2=\mathbbm{1}_{\{|A_{T_t^c}-A_{T_t^c}^R|\ge \eta\}}.
\end{align*}
Then we have
\begin{align*}
& \mathbb{E}_P^{0}\Biggl[
\mathbbm{1}_{\{|T_t^c-t_{\alpha}|\le \varepsilon\}}
\int_{\mathbb{R}^d}
\mathbb{E}_{\mathcal{W}}^{x}\!\left[
|g(B_{t_{\alpha}})|^2
\right]
\mathbb{E}_{\mathcal{W}}^{x}\!\left[
\bigl|e^{-iA_{T_t^c}}-e^{-iA_{T_t^c}^R}\bigr|^2
\right]
dx
\Biggr]
\\
&=
\mathbb{E}_P^{0}\Biggl[
\mathbbm{1}_{\{|T_t^c-t_{\alpha}|\le \varepsilon\}}
\int_{\mathbb{R}^d}
\mathbb{E}_{\mathcal{W}}^{x}\!\left[
|g(B_{t_{\alpha}})|^2
\right]
\mathbb{E}_{\mathcal{W}}^{x}\!\left[
\bigl|e^{-iA_{T_t^c}}-e^{-iA_{T_t^c}^R}\bigr|^2
\chi_{R,c}^1
\right]
dx
\Biggr]
\\
& \quad +
\mathbb{E}_P^{0}\Biggl[
\mathbbm{1}_{\{|T_t^c-t_{\alpha}|\le \varepsilon\}}
\int_{\mathbb{R}^d}
\mathbb{E}_{\mathcal{W}}^{x}\!\left[
|g(B_{t_{\alpha}})|^2
\right]
\mathbb{E}_{\mathcal{W}}^{x}\!\left[
\bigl|e^{-iA_{T_t^c}}-e^{-iA_{T_t^c}^R}\bigr|^2
\chi_{R,c}^2
\right]
dx
\Biggr].
\end{align*}
It follows that
\begin{align*}
&
\mathbb{E}_P^{0}\Bigl[
\mathbbm{1}_{\{|T_t^c-t_{\alpha}|\le{\varepsilon}\}}
\int_{\mathbb{R}^d}
\mathbb{E}_{\mathcal{W}}^{x}\!\Bigl[
|g(B_{t_{\alpha}})|^2
\Bigr]
\mathbb{E}_{\mathcal{W}}^{x}\!\Bigl[
\bigl|e^{-iA_{T_t^c}}-e^{-iA_{T_t^c}^R}\bigr|^2
\chi_{R,c}^1
\Bigr]
dx
\Bigr]
\\
&\le
\mathbb{E}_P^{0}\Bigl[
\mathbbm{1}_{\{|T_t^c-t_{\alpha}|\le{\varepsilon}\}}
\int_{\mathbb{R}^d}
\mathbb{E}_{\mathcal{W}}^{x}\!\Bigl[
|g(B_{t_{\alpha}})|^2
\Bigr]
\mathbb{E}_{\mathcal{W}}^x[{\eta}^2]
dx
\Bigr]
\le
{\norm{g}}^2{\eta}^2.
\end{align*}
Next, we evaluate the term which contains $\chi_{R,c}^2$. First, there exists sufficiently large $R'>0$ such that for all $R>R'$,

\begin{equation*}
\left(\frac{2}{\sqrt{2\pi (t_{\alpha}-\varepsilon)}} \int_R^{\infty} e^{-\frac{y^2}{2(t_{\alpha}-\varepsilon)}} \, dy \right)^d
\le \frac{3\eta^2}{16}.
\end{equation*}
Hence for all $R>R^{\prime}$, we have
\begin{align*}
&
\mathbb{E}_P^{0}\Bigl[
\mathbbm{1}_{\{|T_t^c-t_{\alpha}|\le{\varepsilon}\}}
\int_{\mathbb{R}^d}
\mathbb{E}_{\mathcal{W}}^{x}\!\Bigl[
|g(B_{t_{\alpha}})|^2
\Bigr]
\mathbb{E}_{\mathcal{W}}^{x}\!\Bigl[
\bigl|e^{-iA_{T_t^c}}-e^{-iA_{T_t^c}^R}\bigr|^2
\chi_{R,c}^2
\Bigr]
dx
\Bigr]
\\
&\le
\mathbb{E}_P^{0}\Bigl[
\mathbbm{1}_{\{|T_t^c-t_{\alpha}|\le{\varepsilon}\}}
\int_{\mathbb{R}^d}
\mathbb{E}_{\mathcal{W}}^{x}\!\Bigl[
|g(B_{t_{\alpha}})|^2
\Bigr]
\mathbb{E}_{\mathcal{W}}^{x}\!\Bigl[
4\chi_{R,c}^2
\Bigr]
dx
\Bigr]
\\
&\le
16\mathbb{E}_P^{0}\Bigl[
\mathbbm{1}_{\{|T_t^c-t_{\alpha}|\le{\varepsilon}\}}
\int_{\mathbb{R}^d}
\mathbb{E}_{\mathcal{W}}^{x}\!\Bigl[
|g(B_{t_{\alpha}})|^2
\Bigr]
dx
\left(
\frac{2}{\sqrt{2{\pi}{T_t^c}}}
\displaystyle\int_R^{\infty}
e^{-\frac{y^2}{2{T_t^c}}}dy
\right)^d
\Bigr]
\\
&\le
16\mathbb{E}_P^{0}\Bigl[
\mathbbm{1}_{\{|T_t^c-t_{\alpha}|\le{\varepsilon}\}}
\int_{\mathbb{R}^d}
\mathbb{E}_{\mathcal{W}}^{x}\!\Bigl[
|g(B_{t_{\alpha}})|^2
\Bigr]
dx\Bigr]
\left(
\frac{2}{\sqrt{2{\pi}({t_{\alpha}-\varepsilon})}}
\displaystyle\int_R^{\infty}
e^{-\frac{y^2}{2({t_{\alpha}-\varepsilon})}}dy
\right)^d
\le
3{\norm{g}}^2{\eta}^2 .
\end{align*}

In the second inequality, we used (3.3) of Lemma 3.7. 
Hence for all $R>R^{\prime}$, the second term of the most right-hand side of (3.5) is dominated as
\begin{align}
&{\norm{f}}^2_{L^2}\mathbb{E}_P^{0}\Bigl[
\mathbbm{1}_{\{|T_t^c-t_{\alpha}|\le\varepsilon\}}
\int_{\mathbb{R}^d}\mathbb{E}_{\mathcal{W}}^{x}\!\Bigl[
|g(B_{t_{\alpha}})|^2\Bigr]
\mathbb{E}_{\mathcal{W}}^{x}\!\Bigl[
\bigl|e^{-iA_{T_t^c}}-e^{-iA_{T_t^c}^R}\bigr|^2
\Bigr]dx
\Bigr] \nonumber\le
4{\norm{f}}^2_{L^2}{\norm{g}}^2\eta^2.
\end{align}

Thus the first term of the most right-hand side of (3.4) can be
\begin{align}
&\left(\int_{\mathbb{R}^d}\mathbb{E}_{\mathcal{W},P}^{x,0}\!\Bigl[
|f(B_0)||g(B_{t_{\alpha}})|
\bigl|e^{-iA_{T_t^c}}-e^{-iA_{T_t^c}^R}\bigr|
\Bigr]dx\right)^2 \nonumber\le4{\norm{f}}^2_{L^2}{\norm{g}}^2\eta^2+4{\norm{f}}^2_{L^2}{\norm{g}}^2_{L^2}
\mathbb{E}_P^{0}\Bigl[
\mathbbm{1}_{\{|T_t^c-t_{\alpha}|\ge\varepsilon\}}
\Bigr].
\end{align}

We can also see that there exists $R^{\prime\prime}$ such that the third term of the most right-hand side of (3.4) satisfies for all $R>R^{\prime\prime}$,

\begin{equation}
\int_{\mathbb{R}^d}\mathbb{E}_{\mathcal{W},P}^{x,0}\!\Bigl[
|f(B_0)||g(B_{t_{\alpha}})|
\bigl|e^{-iA_{t_{\alpha}}^R}-e^{-iA_{t_{\alpha}}}\bigr|
\Bigr]dx\le{\norm{f}}_{L^2}{\norm{g}}_{L^2}\eta\;.
\end{equation}
Let $R>\max\{R^{\prime},R^{\prime\prime}\}$. Then we have

\begin{equation}
\hspace{-1em}
\begin{aligned}
&\int_{\mathbb{R}^d}\mathbb{E}_{\mathcal{W},P}^{x,0}\!\Bigl[
|f(B_0)||g(B_{t_{\alpha}})|
\bigl|e^{-i\int_0^{T_t^c} a(B_s)\circ dB_s}
      -e^{-i\int_0^{t_{\alpha}} a(B_s)\circ dB_s}\bigr|
\Bigr]dx \\[0.5em]
&\qquad\le 2\norm{f}_{L^2}\norm{g}_{L^2}\left({\eta}^2+\mathbb{E}_P^{0}\Bigl[
\mathbbm{1}_{\{|T_t^c-t_{\alpha}|\ge\varepsilon\}}
\Bigr]\right)^{\frac{1}{2}} \\[0.5em]
&\qquad\qquad+
\int_{\mathbb{R}^d}\mathbb{E}_{\mathcal{W},P}^{x,0}\!\Bigl[
|f(B_0)||g(B_{t_{\alpha}})|
\bigl|e^{-i{A_{T_t^c}^R}}-e^{-iA_{t_{\alpha}}^R}\bigr|
\Bigr]dx+{\norm{f}}_{L^2}{\norm{g}}_{L^2}\eta. \\[0.5em]
\end{aligned}
\refstepcounter{equation}
\tag*{\hspace{2em}\raisebox{-0.5\baselineskip}{(\theequation)}}
\end{equation}

By Lemma 3.6, the second term of the right-hand side of (3.7) converges to $0$ as $c\to{\infty}$. Hence we can obtain that
\begin{align*}
&\lim_{c\to{\infty}}\int_{\mathbb{R}^d}
\mathbb{E}_{\mathcal{W},P}^{x,0}\!\Bigl[
|f(B_0)|\,|g(B_{t_{\alpha}})|
\bigl|
e^{-i\int_0^{T_t^c} a(B_s)\circ dB_s}
-
e^{-i\int_0^{t_{\alpha}} a(B_s)\circ dB_s}
\bigr|
\Bigr]dx\le
3\,{\norm{f}}_{L^2}\,{\norm{g}}_{L^2}\,\eta\;.
\end{align*}

Thus the theorem is proved.
\end{proof}
\end{theorem}

\section{Non-relativistic limit of generalized Pauli operators}

In this section, we establish the non-relativistic limit of generalized Pauli operators. First, we verify the self-adjointness of generalized Pauli operators. We then derive a path integral representation and investigate the non-relativistic limit.

\subsection{Pauli operator}

We define Pauli operators as operators acting on $L^2(\mathbb{R}^3;\mathbb{C}^2)$ under a strong restriction:
\begin{align*}
\text{(A.3)}\quad a \in C^2_b(\mathbb{R}^3), \quad V\in{C_b(\mathbb{R}^3)}. 
\end{align*}

The $2\times2$ Pauli matrices are defined by

\[
\sigma_1=
\begin{pmatrix}
0 & 1 \\
1 & 0
\end{pmatrix},\quad
\sigma_2=
\begin{pmatrix}
0 & -i \\
i & 0
\end{pmatrix},\quad
\sigma_3=
\begin{pmatrix}
1 & 0 \\
0 & -1
\end{pmatrix}.
\]

We put $\sigma=(\sigma_1,\sigma_2,\sigma_3)$. They satisfy the anti-commutation relations:

\begin{equation}
\{\sigma_{\mu},\sigma_{\nu}\}=2\delta_{\mu\nu}
,\quad \mu,\nu=1,2,3.
\end{equation}

This leads to the formula as below:
\begin{equation}
\{\sigma_{\mu},\sigma_{\nu}\}=2i\sum_{\lambda=1}^3\varepsilon^{\mu\nu\lambda}\sigma_{\lambda},
\end{equation}
where $\varepsilon^{\mu\nu\lambda}$ is the Levi-Cività tensor given by

\[
\varepsilon^{\mu\nu\lambda}=
\begin{cases}
1, & \mu\nu\lambda\:\text{is}\:\text{an}\:\text{even}\:\text{permutation}\:\text{of}\:123, \\
-1,  &\mu\nu\lambda\:\text{is}\:\text{an}\:\text{odd}\:\text{permutation}\:\text{of}\:123, \\
0, &\text{otherwise}.
\end{cases}
\]

Suppose that the vector potential $a$ satisfies (A.3).
The Pauli operator 
is defined by
\begin{equation}
H_{\mathrm{S}}(a)
= \frac{1}{2}\bigl(\sigma \cdot (-i\nabla - a)\bigr)^2 + V.
\end{equation}
Note that $H_{\mathrm{S}}(a)$ is a self-adjoint operator on $D(-\Delta)$. We set $h_0(a)=\frac{1}{2}\bigl(-i\nabla - a\bigr)^2$ and ${\nabla}\times{a}=b=(b_1,b_2,b_3)$. By using the relation (4.2), we can obtain the equality:

\begin{equation}
H_{\mathrm{S}}(a)=h_0(a)+V-\frac{1}{2}\sigma\cdot b.
\end{equation}

\quad To construct a path integral representation of $e^{-tH_{\mathrm{S}}(a)}$,
we transform $H_{\mathrm{S}}(a)$ to an operator acting on
$\mathbb{C}$-valued $L^2$-functions.
Consider the isomorphism
\[
U : L^2(\mathbb{R}^3;\mathbb{C}^2)
   \longrightarrow L^2(\mathbb{R}^3 \times \mathbb{Z}_2)
\]
defined by

\[
L^2(\mathbb{R}^3;\mathbb{C}^2)\ni
\begin{pmatrix}
f(x,+1)
\\
f(x,-1)
\end{pmatrix}
\mapsto
U\begin{pmatrix}
f(x,+1)
\\
f(x,-1)
\end{pmatrix}
=f(x,\theta)
\in
L^2(\mathbb{R}^3\times{\mathbb{Z}_2})
.\]

Suppose (A.3) and $b\in(L^{\infty}(\mathbb{R}^3))^3$.
The Pauli operators on $L^2(\mathbb{R}^3\times{\mathbb{Z}_2})$ is defined by

\begin{equation}
H_{\mathbb{Z}_2}(a,b)=UH_{\mathrm{S}}(a)U^{-1}.
\end{equation}

$H_{\mathbb{Z}_2}(a,b)$ can be explicitly given by 

\begin{equation}
\biggr(H_{\mathbb{Z}_2}(a,b)f\biggl)(x,\theta)=\Bigl(h_0(a)+V-\frac{1}{2}b_3(x)\Bigl)f(x,\theta)-\frac{1}{2}(b_1(x)-i\theta{b_2(x)})f(x,-\theta).
\end{equation}

\quad Next, we define $H_{\mathbb{Z}_2}(a,b)$ with singular vector potentials. 
We redefine $H_{\mathbb{Z}_2}(a,b)$ with $h_0(a)$ replaced by $h(a)$ defined in Proposition 3.2. We introduce assumption (A.4):
\begin{align*}
\text{(A.4)}\quad & a \in (L^2_{\mathrm{loc}}(\mathbb{R}^3))^3, 
\quad \nabla \cdot a \in L^1_{\mathrm{loc}}(\mathbb{R}^3), \quad  V\in{C_b(\mathbb{R}^3)}.
\end{align*}

\begin{definition}Suppose (A.4) and $b\in(L^{\infty}(\mathbb{R}^3))^3$. The Pauli operator with a vector potential $a$ is defined by
\begin{equation}
\begin{aligned}
D(H_{\mathbb{Z}_2}(a,b)) &= D(h(a)), \\
\bigl(H_{\mathbb{Z}_2}(a,b)f\bigr)(x,\theta) &= \left(h(a) + V - \frac{1}{2}b_3(x)\right) f(x,\theta) - \frac{1}{2}(b_1(x) - i\theta b_2(x)) f(x,-\theta),
\end{aligned}
\end{equation}
where $f \in D(H_{\mathbb{Z}_2}(a,b))$.
\end{definition}
Since $b$ and $V$ are bounded multiplication operators, $H_{\mathbb{Z}_2}(a,b)$ is a self-adjoint operator on $D(H_{\mathbb{Z}_2}(a,b))$ and bounded from below. 

\quad Finally, we derive a path integral representation of $e^{-tH_{\mathbb{Z}_2}(a,b)}$. To this end, it is necessary to introduce an additional random process describing the spin component. 

\begin{definition}
Let $({\Omega}',{\mathcal{F}}',Q)$ be a probability space, and $(N_t)_{t\ge0}$ be a Poisson process with  intensity $1$ on $({\Omega}',{\mathcal{F}}',Q)$. The random process $(\theta_t)_{t\ge0}$ is defined by 
\[
\theta_t
=(-1)^{N_t}
,\quad t\ge0.
\]
The random process $(\theta_t)_{t\ge0}$ is called a spin process.
\end{definition}

We define a random process $(q_t)_{t\ge0}$ on $({\mathscr{X}}\times{\Omega}',\mathcal{B}(\mathscr{X})\otimes{{\mathcal{F}}'},{\mathcal{W}^x}\otimes{Q})$
by
\begin{equation}
q_t=(B_t,\theta_t), \quad t\ge0.   
\end{equation}

With these preparations, we can derive a path integral representation of $e^{-tH_{\mathbb{Z}_2}(a,b)}$. 
We set 
\[\Pi_s(y) = (2\pi s)^{-3/2} e^{-\frac{|y|^2}{2s}}.\] 
\begin{proposition}
Suppose that (A.4) and $b\in(L^{\infty}(\mathbb{R}^3))^3$. In addition, we assume that
\begin{equation}
\displaystyle\int_0^t{ds}\displaystyle\int_{\mathbb{R}^3}\left|\log\frac{1}{2}\sqrt{{b_1(y)}^2+{b_2(y)}^2}\right|\Pi_s(y-x)dy<\infty
\end{equation}
for all $(x,t)\in{\mathbb{R}^3\times[0,\infty)}$. Then, for $f,g\in{L^2(\mathbb{R}^3\times{\mathbb{Z}_2})}$, we have

\begin{equation}
(f,e^{-tH_{\mathbb{Z}_2}(a,b)}g)=e^t\sum_{\sigma=1,2}\displaystyle\int_{\mathbb{R}^3}\mathbb{E}_{\mathcal{W},Q}^{x,\sigma}[\overline{f(q_0)}g(q_{t})e^{Z_t}]dx,\    
\end{equation}
where 
\begin{align*}
Z_t
&= - i \int_0^t a(B_s)\circ \mathrm dB_s
   - \int_0^t V(B_s)\,\mathrm ds
   + \int_0^t\frac{1}{2}\theta_s b_3(B_s)\,\mathrm ds \\
&\quad + \int_0^{t+}
   \log\!\Bigl(
     \tfrac12\bigl(b_1(B_s)- i\theta_s b_2(B_s)\bigr)
   \Bigr)\mathrm dN_s .
\end{align*}
\begin{proof}
Formula (4.10) under (A.3) is proved in \cite[Theorem 4.235]{JHV2020}, also, the formula under (A.4) is proved in \cite[Theorem 4.240]{JHV2020}.
\end{proof}
\end{proposition}

We write $U(x,\theta)=\frac{1}{2}{\theta}b_3(x)$ and $W(x,-\theta)=\log\!\Bigl(
     \frac{1}{2}\bigl(b_1(x)- i\theta b_2(x)\bigr)
   \Bigr)$.

\begin{remark}
To prove (4.10), we need to check that $\left|\displaystyle\int_{0}^{t+}W(B_s,-{\theta}_s)dNs\right|$
is bounded almost surely. Suppose that (4.9) is satisfied. Then we can see that
\begin{align*}
&\left|
\mathbb{E}_{\mathcal{W},Q}^{x,\sigma}
\left[
\int_0^{t+} W(B_s,-\theta_s)\,\mathrm dN_s
\right]
\right|
\\
&{\le}\displaystyle\int_0^t{ds}\displaystyle\int_{\mathbb{R}^3}\left|\log\Bigl(\frac{1}{2}\sqrt{{b_1(y)}^2+{b_2(y)}^2}\Bigr)\right|(\sqrt{2{\pi}s})^{-\frac{3}{2}}\displaystyle\int_{\mathbb{R}^3}e^{-\frac{|x-y|^2}{2s}}dy<\infty.
\end{align*}
\end{remark}

\subsection{Generalized relativistic Pauli operators}

In this subsection, we define generalized Pauli operators with singular vector potentials.
We denote the operator $H_{\mathbb{Z}_2}(a,b)$ with $V=0$ by $H^{0}_{\mathbb{Z}_2}(a,b)$. 
Suppose (A.4) and $b \in (L^{\infty}(\mathbb{R}^3))^3$. 
We can define the Hamiltonian with the kinetic term $h(a)$. 
The resulting operator on $L^2(\mathbb{R}^3 \times \mathbb{Z}_2)$ is given as follows.

\begin{definition}
Suppose (A.4) and $b\in(L^{\infty}(\mathbb{R}^3))^3$. Let $\Psi_c^{\alpha}$ be the Bernstein function by (2.3). Then generalized Pauli operators with singular vector potentials are defined by

\begin{equation}
\begin{gathered}
D(H^{c}_{\alpha,{\mathbb{Z}}_2}(a,b)) = D(\Psi_c^{\alpha}(h(a))), \\[8pt]
\begin{aligned}
H^{c}_{\alpha,{\mathbb{Z}}_2}(a,b)=\Psi_c^{\alpha}(H^{0}_{\mathbb{Z}_2}(a,b))+V.
\end{aligned}
\end{gathered}
\end{equation}
\end{definition}

We can see that $H^{c}_{\alpha,{\mathbb{Z}}_2}(a,b)$ is self-adjoint on $D(\Psi_c^{\alpha}(H^{0}_{\mathbb{Z}_2}(a,b))$ and bounded from below.

Let $(T_t^c)_{t\ge0}$ be the $\alpha/2$-relativistic subordinator, and $\theta_t$ be the spin process on $({\Omega}',{\mathcal{F}}',Q)$. We consider the random process $({{q}}_{T_t^c})_{t\ge0}$ on a probability space $({\mathscr{X}}\times{\Omega}\times{\Omega}',\mathcal{B}(\mathscr{X})\otimes{\mathcal{F}}\otimes{{\mathcal{F}}'},{\mathcal{W}^x}\otimes{P}\otimes{Q})$ defined by
\begin{equation}
({{q}}_{T_t^c})_{t\ge0}=(B_{T_t^c},{\theta}_{T_t^c})_{t\ge0}.  
\end{equation}

By this random process, we can derive the path integral representation of $e^{-tH^{c}_{\alpha,{\mathbb{Z}}_2}}$. 

\begin{proposition}
Suppose (A.4) and $b\in(L^{\infty}(\mathbb{R}^3))^3$. In addition, assume that $b$ satisfies (4.9).
Then, for $f,g\in{L^2(\mathbb{R}^3\times{\mathbb{Z}_2})}$, we have

\begin{equation}
(f,e^{-tH^{c}_{\alpha,{\mathbb{Z}}_2}(a,b)}g)=\sum_{\sigma=1,2}\displaystyle\int_{\mathbb{R}^3}\mathbb{E}_{\mathcal{W},P,Q}^{x,0,\sigma}[e^{T_t^c}\overline{f(q_0)}g({q}_{T_t^c})e^{{\tilde{Z_t^c}}}]dx,\    
\end{equation}

where 
\begin{align*}
\tilde{Z_t^c}
&= - i \int_0^{T_t^c} a(B_s)\circ \mathrm dB_s
   - \int_0^t V(B_{T_s^c})\,\mathrm ds
   + \int_0^{T^c_t}\frac{1}{2}\theta_s b_3(B_s)\,\mathrm ds \\
&\quad + \int_0^{T_t^c+}
   \log\!\Bigl(
     \tfrac12\bigl(b_1(B_s)- i\theta_s b_2(B_s)\bigr)
   \Bigr)\mathrm dN_s.
\end{align*}
\begin{proof}
The proof is a minor modification of \cite[Theorem 4.249]{JHV2020}.   
\end{proof}
\end{proposition}

\subsection{Non-relativistic limit}
We establish the non-relativistic limit of generalized Pauli operators in this subsection.
We write
\begin{align*}
\frac{\alpha}{m^{\frac{2}{\alpha}-1}}H_{\mathbb{Z}_2}(a,b)    
=H_{\alpha,\mathbb{Z}_2}(a,b).
\end{align*}
Now we are in the position to state the main theorem in this paper.
\begin{theorem}
Suppose (A.4) and $b\in(L^{\infty}(\mathbb{R}^3))^3$. In addition, assume that $b$ satisfies (4.9).
Then we have 
\begin{equation}
s-\lim_{c\to{\infty}}e^{-tH^{c}_{\alpha,{\mathbb{Z}}_2}(a,b)}=e^{-tH_{\alpha,{\mathbb{Z}}_2}(a,b)}\quad \text{for}\; \text{all}\;t\ge0.
\end{equation}
\end{theorem}

\begin{proof}
 We will show that $\|e^{-tH^{c}_{\alpha,{\mathbb{Z}}_2}(a,b)}\|$ is uniformly bounded for $c>0$ in Lemma 4.8 below. Then by the limiting argument, it suffices to show that 
\[
\lim_{c\to{\infty}}(f,e^{-tH^{c}_{\alpha,{\mathbb{Z}}_2}(a,b)}g)=(f,e^{-tH_{\alpha,{\mathbb{Z}}_2}(a,b)}g)
\]
for all $f,g\in{C_0^{\infty}(\mathbb{R}^3\times{\mathbb{Z}_2})}$. 
We have
\begin{equation}
|(f,e^{-tH^{c}_{\alpha,{\mathbb{Z}}_2}(a,b)}g)-(f,e^{-tH_{\alpha,{\mathbb{Z}}_2}(a,b)}g)|
\le{S_1^c}+{S_2^c}+{S_3^c}.
\end{equation}
Here we write
\begin{align}
S_1^c &= \sum_{\sigma=1,2}\int_{\mathbb{R}^3}|f(x,\sigma)|\mathbb{E}_{\mathcal{W},P,Q}^{x,0,\sigma}\bigl[|g({q}_{T_t^c})||e^{\tilde{Z_t^c}}| |e^{T_t^c}-e^{t_{\alpha}}| \bigr]dx, \\[10pt]
S_2^c &= \sum_{\sigma=1,2}\int_{\mathbb{R}^3}|f(x,\sigma)|\mathbb{E}_{\mathcal{W},P,Q}^{x,0,\sigma}\bigl[e^{t_{\alpha}}|g({q}_{T_t^c})-g(q_{t_{\alpha}})||e^{\tilde{Z_{t}^c}}|\bigr]dx, \\[10pt]
S_3^c &= \sum_{\sigma=1,2}\int_{\mathbb{R}^3}|f(x,\sigma)|\mathbb{E}_{\mathcal{W},P,Q}^{x,0,\sigma}\bigl[e^{t_\alpha}|g({q}_{t_{\alpha}})||e^{\tilde{Z_t^c}}-e^{Z_{t\alpha}}|\bigr]dx.
\end{align}
\allowdisplaybreaks
The convergence of $S_1^c$ to $0$ as $c\to{\infty}$ can be shown as follows:
\begin{align*}
|S_1^c|^2
&\le \|f\|_{L^2}^2
\sum_{\sigma=1,2}\int_{\mathbb{R}^3}
\mathbb{E}_{\mathcal{W},P,Q}^{x,0,\sigma}
\Bigl[|e^{\tilde{Z_t^c}}|^2
|e^{T_t^c}-e^{t_{\alpha}}|^2\Bigr]
\mathbb{E}_{\mathcal{W},P,Q}^{x,0,\sigma}
\Bigl[|g(q_{T_t^c})|^2\Bigr]
\,dx \\
&\le \|f\|_{L^2}^2
\sum_{\sigma=1,2}\int_{\mathbb{R}^3}
\mathbb{E}_{P}^{0}\Bigl[
e^{2t\|V\|_{\infty}}
e^{2M T_t^c}
|e^{T_t^c}-e^{t_{\alpha}}|^2
\mathbb{E}_{Q}^{\sigma}\!\bigl[(M')^{2N_{T_t^c}}\bigr]
\Bigr] 
\mathbb{E}_{\mathcal{W},P,Q}^{x,0,\sigma}
\Bigl[|g(q_{T_t^c})|^2\Bigr]\,dx \\
&= \|f\|_{L^2}^2
\sum_{\sigma=1,2}\int_{\mathbb{R}^3}
\mathbb{E}_{P}^{0}\Bigl[
e^{2t\|V\|_{\infty}}
e^{2M T_t^c}
|e^{T_t^c}-e^{t_{\alpha}}|^2
 e^{((M')^2-1)T_t^c}
\Bigr] 
\mathbb{E}_{\mathcal{W},P,Q}^{x,0,\sigma}
\Bigl[|g(q_{T_t^c})|^2\Bigr]\,dx \\
&\le\|f\|_{L^2}^2 \|g\|_{L^2}^2
\sup_{c>0}
\Bigl(
\mathbb{E}_{P}^{0}
[e^{4t\|V\|_{\infty}} e^{4M T_t^c}]
\Bigr)^{1/2}
\sup_{c>0}
\Bigl(
\mathbb{E}_{P}^{0}
[e^{2((M')^2-1)T_t^c}]
\Bigr)^{1/2}
\Bigl(
\mathbb{E}_{P}^{0}
[|e^{T_t^c}-e^{t_{\alpha}}|^4]
\Bigr)^{1/2}.
\end{align*}

Here $M=\displaystyle \sup_{x\in\mathbb{R}^3}|b_3(x)|$ and $M'=\displaystyle \sup_{x\in\mathbb{R}^3}\frac{1}{2}\sqrt{b_1(x)^2+b_2(x)^2}$. We also denote the norm on $L^2(\mathbb{R}^3\times{\mathbb{Z}_2})$ by $\|{\cdot}\|_{L^2}$. Hence $S_1^c$ converges to $0$ as $c\to{\infty}$ by Lemma 2.4. We will evaluate $S_2^c$. By the Schwarz inequality, we have
\begin{align*}
|S_2^c|^2
&\le
\left(
\sum_{\sigma=1,2}\int_{\mathbb{R}^3}
|f(x,\sigma)|
\mathbb{E}_{\mathcal{W},P,Q}^{x,0,\sigma}
\bigl[
e^{Z_{T_t^c}}
|g(q_{T_t^c})-g(q_{t_\alpha})|
\bigr]
\,dx
\right)^2 \\
&\le \|f\|_{L^2}^2
\sum_{\sigma=1,2}\int_{\mathbb{R}^3}
\mathbb{E}_{\mathcal{W},P,Q}^{x,0,\sigma}
\bigl[e^{2Z_{T_t^c}}\bigr]
\mathbb{E}_{\mathcal{W},P,Q}^{x,0,\sigma}
\bigl[|g(q_{T_t^c})-g(q_{t_\alpha})|^2\bigr]
\,dx \\
&\le \|f\|_{L^2}^2
\sum_{\sigma=1,2}\int_{\mathbb{R}^3}
\mathbb{E}_{P}^{0}\Bigl[
e^{2t\|V\|_{\infty}}
e^{2M T_t^c}
e^{((M')^2-1)T_t^c}
\Bigr]  
\mathbb{E}_{\mathcal{W},P,Q}^{x,0,\sigma}
\bigl[|g(q_{T_t^c})-g(q_{t_\alpha})|^2\bigr]
\,dx \\
&\le
C\|f\|_{L^2}^2
\sum_{\sigma=1,2}\int_{\mathbb{R}^3}
\mathbb{E}_{\mathcal{W},P,Q}^{x,0,\sigma}
\bigl[
|g(q_{T_t^c})-g(q_{t_\alpha})|^2
\mathbbm{1}_{\{T_t^c \ge t_\alpha\}}
\bigr]
\,dx \\
&\quad\quad\quad +
C\|f\|_{L^2}^2
\sum_{\sigma=1,2}\int_{\mathbb{R}^3}
\mathbb{E}_{\mathcal{W},P,Q}^{x,0,\sigma}
\bigl[
|g(q_{T_t^c})-g(q_{t_\alpha})|^2
\mathbbm{1}_{\{T_t^c \le t_\alpha\}}
\bigr]
\,dx,
\end{align*}

where $C=e^{2t{\|V\|}_{\infty}}\displaystyle\sup_{c>0}\Bigl(\mathbb{E}_{P}^{0}\bigl[{e^{2MT_t^c}}e^{((M')^2-1)T_t^c}]\Bigr)$. Note that $\displaystyle\sup_{c>0}\Bigl(\mathbb{E}_{P}^{0}\bigl[{e^{uT_t^c}}]\Bigr)<{\infty}$ for all $u\in\mathbb{R}$ by Lemma 2.4. It suffices to show that the first term on the right-hand side converges to $0$ as $c \to \infty$. By Itô-formula for semimartingales, we have

\begin{align*}
&\sum_{\sigma=1,2}\int_{\mathbb{R}^3}
\mathbb{E}_{\mathcal{W},P,Q}^{x,0,\sigma}
\bigl[|g(q_{T_t^c})-g(q_{t_\alpha})|^2\mathbbm{1}_{\{T_t^c\ge{t_{\alpha}}\}}\bigr]
\,dx\\
&=\mathbb{E}^{0}_{P}\Biggl[\mathbbm{1}_{\{T_t^c\ge{t_{\alpha}}\}}\sum_{\sigma=1,2}\int_{\mathbb{R}^3}
\mathbb{E}_{\mathcal{W},Q}^{x,\sigma}
\bigl[|g(q_{T_t^c})-g(q_{t_\alpha})|^2\bigr]
\,dx\Biggr]
=\sum_{k=1}^{6}I_k^c,
\end{align*}

where
\begin{align*}
I_1^c &= \mathbb{E}^{0}_{P}\Biggl[\mathbbm{1}_{\{T_t^c\ge{t_{\alpha}}\}}\sum_{\sigma=1,2}\int_{\mathbb{R}^3} \mathbb{E}_{\mathcal{W},Q}^{x,\sigma} \Biggl[\Biggl|\int_{t_{\alpha}}^{T_t^c}{\nabla}g(q_s){\cdot}{dB_s}\Biggr|^2\Biggr] \,dx\Biggr], \\[10pt]
I_2^c &= \mathbb{E}^{0}_{P}\Biggl[\mathbbm{1}_{\{T_t^c\ge{t_{\alpha}}\}}\sum_{\sigma=1,2}\int_{\mathbb{R}^3} \mathbb{E}_{\mathcal{W},Q}^{x,\sigma} \Biggl[\Biggl|\int_{t_{\alpha}}^{T_t^c}\frac{1}{2}{\Delta}g(q_s)ds\Biggr|^2\Biggr] \,dx\Biggr], \\[10pt]
I_3^c &= \mathbb{E}^{0}_{P}\Biggl[\mathbbm{1}_{\{T_t^c\ge{t_{\alpha}}\}}\sum_{\sigma=1,2}\int_{\mathbb{R}^3} \mathbb{E}_{\mathcal{W},Q}^{x,\sigma} \Biggl[\Biggl|\int_{t_{\alpha}}^{T_t^c}h(B_s,\theta_s)dN_s\Biggr|^2\Biggr] \,dx\Biggr], \\[10pt]
I_4^c &= 2\mathbb{E}^{0}_{P}\Biggl[\mathbbm{1}_{\{T_t^c\ge{t_{\alpha}}\}} \sum_{\sigma=1,2}\int_{\mathbb{R}^3} \mathbb{E}_{\mathcal{W},Q}^{x,\sigma} \Biggl[\Biggl(\int_{t_{\alpha}}^{T_t^c}{\nabla}g(q_s){\cdot}{dB_s}\Biggr) \Biggl(\int_{t_{\alpha}}^{T_t^c}{\Delta}g(q_r)dr\Biggr)\Biggr] \,dx\Biggr], \\[10pt]
I_5^c &= 2\mathbb{E}^{0}_{P}\Biggl[\mathbbm{1}_{\{T_t^c\ge{t_{\alpha}}\}} \sum_{\sigma=1,2}\int_{\mathbb{R}^3} \mathbb{E}_{\mathcal{W},Q}^{x,\sigma} \Biggl[\Biggl(\int_{t_{\alpha}}^{T_t^c}{\nabla}g(q_s){\cdot}{dB_s}\Biggr) \Biggl(\int_{t_{\alpha}}^{T_t^c}h(B_r,\theta_r)dN_r\Biggr)\Biggr] \,dx\Biggr], \\[10pt]
I_6^c &= 2\mathbb{E}^{0}_{P}\Biggl[\mathbbm{1}_{\{T_t^c\ge{t_{\alpha}}\}} \sum_{\sigma=1,2}\int_{\mathbb{R}^3} \mathbb{E}_{\mathcal{W},Q}^{x,\sigma} \Biggl[\Biggl(\int_{t_{\alpha}}^{T_t^c}{\Delta}g(q_s)ds\Biggr) \Biggl(\int_{t_{\alpha}}^{T_t^c}h(B_r,\theta_r)dN_r\Biggr)\Biggr] \,dx\Biggr].
\end{align*}
Here we set $h(x,\theta)=g(x,-\theta)-g(x,\theta)$. We evaluate $I_1^c$.
By Itô-isometry, we obtain that
\begin{align*}
I_1^c&=
\mathbb{E}^{0}_{P}\Biggl[\mathbbm{1}_{\{T_t^c\ge{t_{\alpha}}\}}\sum_{\sigma=1,2}\int_{\mathbb{R}^3}
\mathbb{E}_{\mathcal{W},Q}^{x,\sigma}
\Biggl[\int_{t_{\alpha}}^{T_t^c}|{\nabla}g(q_s)|^2ds\Biggr]
\,dx\Biggr]\\
&=\mathbb{E}^{0}_{P}\Biggl[\mathbbm{1}_{\{T_t^c\ge{t_{\alpha}}\}}\int_{t_{\alpha}}^{T_t^c}
\Biggl[\sum_{\sigma=1,2}\int_{\mathbb{R}^3}
\mathbb{E}_{\mathcal{W},Q}^{x,\sigma}|{\nabla}g(q_s)|^2dx\Biggr]
\,ds\Biggr]\\
&\le\|{\nabla}g\|_{L_{2}}^2\mathbb{E}^{0}_{P}\Biggl[\mathbbm{1}_{\{T_t^c\ge{t_{\alpha}}\}}{|T_t^c-t_{\alpha}|}
\Biggr]
\to0\quad\quad(c\to{\infty}).
\end{align*}

The convergence of the right-hand side follows from Corollary 2.5. 
The convergence of $I_2^c$ to $0$ as $c\to{\infty}$ can be proved similarly to $I^c_1$. Next, we evaluate $I_3^c$. By the definition of compensated Poisson integrals and Itô-isometry, we have
\begin{align*}
I_3^c
&=
\mathbb{E}^{0}_{P}\Biggl[\mathbbm{1}_{\{T_t^c\ge{t_{\alpha}}\}}
\sum_{\sigma=1,2}\int_{\mathbb{R}^3}
\mathbb{E}_{\mathcal{W},Q}^{x,\sigma}
\Biggl[\Biggl|\int_{t_{\alpha}}^{T_t^c}h(B_s,\theta_s)d\tilde{N}_s
+\int_{t_{\alpha}}^{T_t^c}h(B_s,\theta_s)ds\Biggr|^2\Biggr]
\,dx\Biggr] \\
&\le
2\mathbb{E}^{0}_{P}\Biggl[\mathbbm{1}_{\{T_t^c\ge{t_{\alpha}}\}}
\sum_{\sigma=1,2}\int_{\mathbb{R}^3}
\mathbb{E}_{\mathcal{W},Q}^{x,\sigma} 
\Biggl[\Biggl|\int_{t_{\alpha}}^{T_t^c}h(B_s,\theta_s)d\tilde{N}_s\Biggr|^2
+\Biggl|\int_{t_{\alpha}}^{T_t^c}h(B_s,\theta_s)ds\Biggr|^2
\Biggr]dx\Biggr] \\
&\le
2\mathbb{E}^{0}_{P}\Biggl[\mathbbm{1}_{\{T_t^c\ge{t_{\alpha}}\}}
\sum_{\sigma=1,2}\int_{\mathbb{R}^3}
\mathbb{E}_{\mathcal{W},Q}^{x,\sigma} 
\Biggl[\int_{t_{\alpha}}^{T_t^c}|h(B_s,\theta_s)|^2ds
+|T_t^c-t_{\alpha}|\int_{t_{\alpha}}^{T_t^c}|h(B_s,\theta_s)|^2ds
\Biggr]dx\Biggr] \\
&=2\|h\|_{L_{2}}^2
\mathbb{E}^{0}_{P}\Biggl[
\mathbbm{1}_{\{T_t^c\ge{t_{\alpha}}\}}
\bigl(|T_t^c-t_{\alpha}|+|T_t^c-t_{\alpha}|^2\bigr)
\Biggr]
\to0
\qquad (c\to{\infty}).
\end{align*}

We show the convergence of $I_4^c$ to $0$ as $c\to{\infty}$. By the Schwarz inequality, we have
\begin{align*}
&|I_4^c|^2 \\
&\le4\mathbb{E}^{0}_{P}\Biggl[\mathbbm{1}_{\{T_t^c\ge{t_{\alpha}}\}}
\Biggl(\sum_{\sigma=1,2}\int_{\mathbb{R}^3}
\Biggl(\mathbb{E}_{\mathcal{W},Q}^{x,\sigma} 
\Biggl[\Biggl|\int_{t_{\alpha}}^{T_t^c}{\nabla}g(q_s){\cdot}{dB_s}\Biggr|^2\Biggr]\Biggr)^{\frac{1}{2}}
\Biggl(\mathbb{E}_{\mathcal{W},Q}^{x,\sigma} 
\Biggl[\Biggl|\int_{t_{\alpha}}^{T_t^c}{\Delta}g(q_r)dr\Biggr|^2\Biggr]\Biggr)^{\frac{1}{2}}
\,dx\Biggr)^2\Biggr]\\
&\le4\mathbb{E}^{0}_{P}\Biggl[\mathbbm{1}_{\{T_t^c\ge{t_{\alpha}}\}}
\Biggl(\sum_{\sigma=1,2}\int_{\mathbb{R}^3}
\mathbb{E}_{\mathcal{W},Q}^{x,\sigma} 
\Biggl[\int_{t_{\alpha}}^{T_t^c}|{\nabla}g(q_s)|^2ds\Biggr]\,dx\Biggr)
\Biggl(\sum_{\sigma=1,2}\int_{\mathbb{R}^3}
\mathbb{E}_{\mathcal{W},Q}^{x,\sigma} \Biggl[\Biggl|\int_{t_{\alpha}}^{T_t^c}{\Delta}g(q_r)dr\Biggr|^2\Biggr]
\,dx\Biggr)\Biggr]\\
&\le4\mathbb{E}^{0}_{P}\bigl[\mathbbm{1}_{\{T_t^c\ge{t_{\alpha}}\}}|T_t^c-t_{\alpha}|^3{\|{\nabla}g\|^2_{L^2}}{\|{\Delta}g\|^2_{L^2}}\bigr] \to0
\qquad (c\to{\infty}).
\end{align*}

The convergence of $I_5^c$ to $0$ as $c\to{\infty}$ follows similarly to that of $I_4^c$. We will show that $I^c_6\to0$ as $c\to{\infty}$. We have 
\begin{align*}
|I^c_6|
&\le
\mathbb{E}^{0}_{P}\Biggl[\mathbbm{1}_{\{T_t^c\ge{t_{\alpha}}\}}
\sum_{\sigma=1,2}\int_{\mathbb{R}^3}
\mathbb{E}_{\mathcal{W},Q}^{x,\sigma} 
\Biggl[\Biggl(\int_{t_{\alpha}}^{T_t^c}|{\Delta}g(q_s)|ds\Biggr)
\Biggl(\int_{t_{\alpha}}^{T_t^c}|h(B_r,\theta_r)|dN_r\Biggr)\Biggr]
\,dx\Biggr] \\
&\le
2{\|{\Delta}g\|}_{\infty}
\mathbb{E}^{0}_{P}\Biggl[\mathbbm{1}_{\{T_t^c\ge{t_{\alpha}}\}}
\sum_{\sigma=1,2}\int_{\mathbb{R}^3}
|T_t^c-t_{\alpha}| 
\mathbb{E}_{\mathcal{W},Q}^{x,\sigma}
\Biggl[\int_{t_{\alpha}}^{T_t^c}|h(B_r,\theta_r)|dr\Biggr]
\Biggr] \\
&=
2{\|{\Delta}g\|}_{\infty}{\|h\|}_{L_{1}}
\mathbb{E}^{0}_{P}[|T_t^c-t_{\alpha}|^2]
\to0
\qquad (c\to{\infty}).
\end{align*}

Thus, we proved that $S_c^2\to0$ as $c\to{\infty}$. Finally we show $S_c^3\to{0}$ as $c\to{\infty}$. We simply write each integral by
\begin{align*}
&A_t=\int_0^{t} a(B_s)\circ \mathrm dB_s,\quad V_t=\int_0^{t}V(B_{s})\, ds,\quad  V_t^c= \int_0^{t}V(B_{T_s^c})\, ds,\\ 
&U_t= \int_0^{t}\frac{1}{2}\theta_s b_3(B_s),\quad
W_t=\int_0^{t+}
   \log\!\Bigl(
     \tfrac12\bigl(b_1(B_s)- i\theta_s b_2(B_s)\bigr)
   \Bigr) dN_s.
\end{align*}

Then we have
\[
|S_c^3|\le{e^{t_{\alpha}}J_1^c+e^{t_{\alpha}+t\|V\|_{\infty}}J_2^c+e^{t\|V\|_{\infty}}
e^{t_{\alpha}(M+1)}J_3^c+e^{t\|V\|_{\infty}}
e^{t_{\alpha}M}J_4^c}
,\]

where
\begin{align*}
J_1^c &= \sum_{\sigma=1,2} \int_{\mathbb{R}^3} |f(x,\sigma)| \mathbb{E}_{\mathcal{W},P,Q}^{x,0,\sigma} \Bigl[ |g(q_{t_{\alpha}})| | e^{-V_t^c} - e^{-V_t}| |e^{U_{t_{\alpha}}}| |e^{W_{t_{\alpha}}}| \Bigr] \,dx, \\[10pt]
J_2^c &= \sum_{\sigma=1,2} \int_{\mathbb{R}^3} |f(x,\sigma)| \mathbb{E}_{\mathcal{W},P,Q}^{x,0,\sigma} \bigl[ |g(q_{t_{\alpha}})| |e^{U_{T_t^c}}-e^{U_{t_{\alpha}}}| |e^{W_{T_t^c}}| \bigr] \, dx, \\[10pt]
J_3^c &= \sum_{\sigma=1,2} \int_{\mathbb{R}^3} |f(x,\sigma)| \mathbb{E}_{\mathcal{W},P,Q}^{x,0,\sigma} \bigl[ |g(q_{t_{\alpha}})| |e^{-iA_{T_t^c}}-e^{-iA_{t_{\alpha}}}| |e^{W_{t_{\alpha}}}| \bigr] \, dx, \\[10pt]
J_4^c &= \sum_{\sigma=1,2} \int_{\mathbb{R}^3} |f(x,\sigma)| \mathbb{E}_{\mathcal{W},P,Q}^{x,0,\sigma} \bigl[ |g(q_{t_{\alpha}})| |e^{W_{T_t^c}}-e^{W_{t_{\alpha}}}| \bigr] \, dx.
\end{align*}

We have
\begin{align*}
|J_1^c|^2\le{\tilde{M}}\sum_{\sigma=1,2}
\int_{\mathbb{R}^3}\mathbb{E}_{\mathcal{W},P,Q}^{x,0,\sigma}
\Biggl[|g(q_{t_{\alpha}})|^2\Biggl|\int_0^tV(B_{T_s^c})-V(B_{t_{\alpha}})ds\Biggr|^2\Biggr]
\, dx,   
\end{align*}

where $\tilde{M}
= \|f\|_{L^2}
e^
{2t\|V\|_{\infty}}
e^{t_{\alpha}(2M+((M')^2-1))}$. 
By the dominated convergence theorem, we obtain that $J_1^c\to{0}$ as $c\to{\infty}$. 
We see that
\begin{align*}
|J_2^c|^2&\le
2
\Biggl|
\sum_{\sigma=1,2}
\int_{\mathbb{R}^3}
|f(x,\sigma)|
\mathbb{E}_{\mathcal{W},P,Q}^{x,0,\sigma}
\Bigl[
|g(q_{t_{\alpha}})|
\,|e^{U_{T_t^c}}-e^{U_{t_{\alpha}}}|
\,|e^{W_{T_t^c}}|
\,\mathbbm{1}_{\{T_t^c\ge t_{\alpha}\}}
\Bigr]
\, dx
\Biggr|^2
\\
&\quad+
2
\Biggl|
\sum_{\sigma=1,2}
\int_{\mathbb{R}^3}
|f(x,\sigma)|
\mathbb{E}_{\mathcal{W},P,Q}^{x,0,\sigma}
\Bigl[
|g(q_{t_{\alpha}})|
\,|e^{U_{T_t^c}}-e^{U_{t_{\alpha}}}|
\,|e^{W_{T_t^c}}|
\,\mathbbm{1}_{\{T_t^c\le t_{\alpha}\}}
\Bigr]
\, dx
\Biggr|^2
\\
&\le
(K+K')
\mathbb{E}_{P}^{0}
\bigl[
|T_t^c-t_{\alpha}|^2
\bigr]\longrightarrow 0
\qquad (c\to{\infty}),
\end{align*}

where 
\[
K=2\|f\|^2_{L^2}\|g\|^2_{L^2}\displaystyle\sup_{c>0}\Bigl(\mathbb{E}_{P}^{0}
\bigl[e^{((M')^2-1)T_t^c}e^{2MT_t^c}\bigr]\Bigr),
\]
\[
K'=2e^{2Mt_{\alpha}}\|f\|^2_{L^2}\|g\|^2_{L^2}\displaystyle\sup_{c>0}\Bigl(\mathbb{E}_{P}^{0}
\bigl[e^{((M')^2-1)T_t^c}\bigr]\Bigr).
\]
We show $J_3^c\to{0}$ as $c\to{\infty}$. We have
\begin{align*}
|J_3^c|^2
&\le
\|f\|_{L^1}^2
\sum_{\sigma=1,2}
\int_{\mathbb{R}^3}
|f(x,\sigma)|
\Biggl(
\mathbb{E}_{\mathcal{W},P,Q}^{x,0,\sigma}
\Bigl[
|g(q_{t_{\alpha}})|
\,|e^{-iA_{T_t^c}}-e^{-iA_{t_{\alpha}}}|
\,|e^{W_{t_{\alpha}}}|
\Bigr]
\Biggr)^2
\, dx
\\
&\le
\tilde{K}
\sum_{\sigma=1,2}
\int_{\mathbb{R}^3}
|f(x,\sigma)|
\mathbb{E}_{\mathcal{W},P,Q}^{x,0,\sigma}
\Bigl[
|g(q_{t_{\alpha}})|
\,|e^{-iA_{T_t^c}}-e^{-iA_{t_{\alpha}}}|
\Bigr]
\, dx \to0\quad\quad(c\to{\infty}).
\end{align*}

Here we set $\tilde{K}=2\|g\|_{\infty}{\|f\|}_{L^1}^2 e^{t_{\alpha}((M')^2-1+M)}e^{t\|V\|_{\infty}}$. The convergence of the right-hand side is proved in the same manner as Theorem 3.8.
Thus we can see that $J_3^c\to{0}$ as $c\to{\infty}$. Finally, we have
\begin{align*}
|J_4^c|^2&=\Biggl(\sum_{\sigma=1,2}
\int_{\mathbb{R}^3}
|f(x,\sigma)|
\mathbb{E}_{\mathcal{W},P,Q}^{x,0,\sigma}
\bigl[
|g(q_{t_{\alpha}})|
|e^{W_{T_t^c}}-e^{W_{t_{\alpha}}}|
\bigr]
\, dx\Biggr)^2 \\ 
&\le\|f\|^2_{L^2}
\sum_{\sigma=1,2}\int_{\mathbb{R}^3}
\mathbb{E}_{\mathcal{W},P,Q}^{x,0,\sigma} \left[ \left| e^{W_{T_t^c}} - e^{W_{t_\alpha}} \right|^2 \right]
\mathbb{E}_{\mathcal{W},P,Q}^{x,0,\sigma}
\bigl[|g(q_{t_{\alpha}})|^2
\bigr]
\, dx.
\end{align*}

For all $(x,\sigma)\in{\mathbb{R}^3\times\mathbb{Z}_2}$, we have
\begin{align*}
    &\mathbb{E}_{\mathcal{W},P,Q}^{x,0,\sigma} \left[ \left| e^{W_{T_t^c}} - e^{W_{t_\alpha}} \right|^2 \right] \mathbb{E}_{\mathcal{W},P,Q}^{x,0,\sigma} \left[ \left| g(q_{t_\alpha}) \right|^2 \right] \\
    &\le 2\left( \sup_{c>0} \left( \mathbb{E}_{P}^{0} \left[ e^{((M')^2-1)T_t^c} \right] \right) + e^{2(M'-1)t_\alpha} \right) \mathbb{E}_{\mathcal{W},P,Q}^{x,0,\sigma} \left[ \left| g(q_{t_\alpha}) \right|^2 \right].
\end{align*}
The right-hand side is integrable on ${\mathbb{R}^3}\times{\mathbb{Z}_2}$. We show that $\mathbb{E}_{\mathcal{W},P,Q}^{x,0,\sigma}
\bigl[|e^{W_{T_t^c}}-e^{W_{t_{\alpha}}}|^2
\bigr]\to0$ as $c\to{\infty}$ for all $(x,\sigma)\in{\mathbb{R}^3\times\mathbb{Z}_2}$ in Lemma 4.9 below.
 Hence we can obtain that $J_4^c\to{0}$ as $c\to{\infty}$ by the dominated convergence theorem.
\end{proof}
It remains to show Lemma 4.8 and Lemma 4.9.
\begin{lemma}
For each $t\ge0$, there exists a constant $C_t>0$ independent of $c>0$ such that $\|e^{-tH^{c}_{\alpha,{\mathbb{Z}}_2}(a,b)}\|<C_t$ for all $c>0$.    
\end{lemma}
\begin{proof}
For all $f,g\;{\in}L^2({\mathbb{R}^3}\times{\mathbb{Z}_2})$, we have
\begin{align*}
{|(f,e^{-tH^{c}_{\alpha,{\mathbb{Z}}_2}(a,b)}g)|}^2&\le \|f\|_{L^2}^2
\sum_{\sigma=1,2}\int_{\mathbb{R}^3}
\mathbb{E}_{\mathcal{W},P,Q}^{x,0,\sigma}
\bigl[|e^{2Z_{T_t^c}}|e^{2T_t^c}\bigr]
\mathbb{E}_{\mathcal{W},P,Q}^{x,0,\sigma}
\bigl[|g(q_{T_t^c})|^2\bigr]
\,dx \\
&\le e^{2t{\|V\|}_{\infty}}{\|f\|_{L^2}^2}{\|g\|_{L^2}^2}
\sup_{c>0}\biggl(\mathbb{E}_{P}^{0}
\bigl[e^{4(M+1)T_t^c}\bigr]\biggr)^{\frac{1}{2}}
\sup_{c>0}\biggl(\mathbb{E}_{P}^{0}\bigl[e^{2((M')^2-1)T_t^c}\bigr]\biggr)^{\frac{1}{2}}
\end{align*}
by Proposition 4.6. Setting $C_t= e^{2t{\|V\|_{\infty}}}\displaystyle\sup_{c>0}\biggl(\mathbb{E}_{P}^{0}
\bigl[e^{4(M+1)T_t^c}\bigr]\biggr)^{\frac{1}{2}}
\displaystyle\sup_{c>0}\biggl(\mathbb{E}_{P}^{0}\bigl[e^{2((M')^2-1)T_t^c}\bigr]\biggr)^{\frac{1}{2}}$, we complete the proof.
\end{proof}

\begin{lemma}
For all $(x,\sigma)\in{\mathbb{R}^3\times\mathbb{Z}_2}$, we see that

\begin{equation}
\mathbb{E}_{\mathcal{W},P,Q}^{x,0,\sigma}
\bigl[|e^{W_{T_t^c}}-e^{W_{t_{\alpha}}}|^2
\bigr]\to0\quad (c\to{\infty}).   
\end{equation}
\begin{proof}
We divide the right-hand side of (4.19) into two parts:
\begin{align*}
 &\mathbb{E}_{\mathcal{W},P,Q}^{x,0,\sigma}
\bigl[|e^{W_{T_t^c}}-e^{W_{t_{\alpha}}}|^2
\bigr]=\mathbb{E}_{\mathcal{W},P,Q}^{x,0,\sigma}
\bigl[|e^{W_{T_t^c}}-e^{W_{t_{\alpha}}}|^2
\mathbbm{1}_{|T_t^c-t_{\alpha}|\ge{\varepsilon}}
\bigr]+\mathbb{E}_{\mathcal{W},P,Q}^{x,0,\sigma}
\bigl[|e^{W_{T_t^c}}-e^{W_{t_{\alpha}}}|^2
\mathbbm{1}_{|T_t^c-t_{\alpha}|\le{\varepsilon}}
\bigr].
\end{align*}
We show the convergence of the first term. By the Schwarz inequality, we have
\begin{align*}
&\Biggl(
\mathbb{E}_{\mathcal{W},P,Q}^{x,0,\sigma}
\bigl[
|e^{W_{T_t^c}}-e^{W_{t_{\alpha}}}|^2
\mathbbm{1}_{|T_t^c-t_{\alpha}|\ge{\varepsilon}}
\bigr]
\Biggr)^2
\le{\mathbb{E}}_{\mathcal{W},P,Q}^{x,0,\sigma}
\bigl[
|e^{W_{T_t^c}}-e^{W_{t_{\alpha}}}|^4
\bigr]
{\mathbb{E}}_{P}^{0}
\bigl[
\mathbbm{1}_{|T_t^c-t_{\alpha}|\ge{\varepsilon}}
\bigr]\\
&\le
\displaystyle\sup_{c>0}\Biggl(
\mathbb{E}_{\mathcal{W},P,Q}^{x,0,\sigma}
\bigl[
|e^{W_{T_t^c}}-e^{W_{t_{\alpha}}}|^4
\bigr]
\Biggr)
{\mathbb{E}}_{P}^{0}
\bigl[
\mathbbm{1}_{|T_t^c-t_{\alpha}|\ge{\varepsilon}}
\bigr] \to 0
\end{align*}

as $c\to{\infty}$ by Corollary 2.5. Next, we show the convergence of the second term. By the definition of Poisson integrals, we can see that
\begin{align*}
&\Biggl(\mathbb{E}_{\mathcal{W},P,Q}^{x,0,\sigma}
\bigl[|e^{W_{T_t^c}}-e^{W_{t_{\alpha}}}|^2
\mathbbm{1}_{|T_t^c-t_{\alpha}|\le{\varepsilon}}
\bigr]\Biggr)^2=\Biggl(\mathbb{E}_{\mathcal{W},P,Q}^{x,0,\sigma}
\bigl[|e^{W_{T_t^c}}-e^{W_{t_{\alpha}}}|^2
\mathbbm{1}_{\{|T_t^c-t_{\alpha}|\le{\varepsilon}\}}
\mathbbm{1}_{\{{|N_{T_t^c}-N_{t_{\alpha}}|\ge1}\}}\bigr]\Biggr)^2\\
&\le\displaystyle\sup_{c>0}\Biggl(\mathbb{E}_{\mathcal{W},P,Q}^{x,0,\sigma}\
\bigl[|e^{W_{T_t^c}}-e^{W_{t_{\alpha}}}|^4
\bigr]\Biggr)
\mathbb{E}_{P,Q}^{0,\sigma}
\bigl[ \mathbbm{1}_{\{ \lvert T_t^c - t_{\alpha} \rvert \le \varepsilon \}}
  \mathbbm{1}_{\{ \lvert N_{T_t^c} - N_{t_{\alpha}} \rvert \ge 1 \}}
\bigr].
\end{align*}

Furthermore, we have
\begin{align*}
&\mathbb{E}_{P,Q}^{0,\sigma}
\bigl[ \mathbbm{1}_{\{ \lvert T_t^c - t_{\alpha} \rvert \le \varepsilon \}}
  \mathbbm{1}_{\{ \lvert N_{T_t^c} - N_{t_{\alpha}} \rvert \ge 1 \}}
\bigr]\le\mathbb{E}_{P,Q}^{0,0}
\bigl[ \mathbbm{1}_{\{ \lvert T_t^c - t_{\alpha} \rvert \le \varepsilon \}}
  \mathbbm{1}_{\{ \lvert N_{T_t^c} - N_{t_{\alpha}} \rvert \ge 1 \}}
\bigr]\\
&\le\mathbb{E}_{P,Q}^{0,0}
\bigl[
\lvert N_{T_t^c} - N_{t_{\alpha}} \rvert 
\mathbbm{1}_{\{ \lvert T_t^c - t_{\alpha} \rvert \le \varepsilon \}}
  \mathbbm{1}_{\{ \lvert N_{T_t^c} - N_{t_{\alpha}} \rvert \ge 1 \}}
\bigr]\le\biggl(2\mathbb{E}_{Q}^{0}
\bigl[N_{\varepsilon}^2
\bigr]\biggr)^{\frac{1}{2}}
=\sqrt{2}{\varepsilon}.
\end{align*}

Hence we obtain that

\[\lim_{c\to{\infty}}\Biggl(\mathbb{E}_{\mathcal{W},P,Q}^{x,0,\sigma}
\bigl[|e^{W_{T_t^c}}-e^{W_{t_{\alpha}}}|^2
\bigr]\Biggl)^2\le
\sqrt{2}\displaystyle\sup_{c>0}\Biggl(\mathbb{E}_{\mathcal{W},P,Q}^{x,0,\sigma}
\bigl[|e^{W_{T_t^c}}-e^{W_{t_{\alpha}}}|^4
\bigr]\Biggr){\varepsilon}.
\]

Since $\varepsilon$ is arbitrary, (4.19) follows.

\end{proof}

\end{lemma}


\begin{thebibliography}{99}

\bibitem{GGM1983}
G.~F.~De Angelis, G.~Jona-Lasinio and M.~Sirugue,
Probabilistic solution of Pauli-type
equations, \textit{J. Phys. A}
\textbf{16} (1983) 2433--2444.

\bibitem{GAM1991}
G.~F.~De Angelis, A.~Rinaldi and M.~Serva,
Imaginary-time path integral for a relativistic spin-(1/2) particle in a magnetic field, \textit{Europhys. Lett.}
\textbf{14} (1991) 95--100.

\bibitem{App2009}
D.~Applebaum,
\textit{Lévy Processes and Stochastic Calculus}, 2nd ed,
Cambridge University Press, 2009.

\bibitem{Bog2009}
K.~Bogdan, T.~Byczkowski, M.~Ryznar, R.~Song, and Z.~Vondra\v{c}ek,
\textit{Potential Analysis of Stable Processes and Its Extensions},
Lecture Notes in Mathematics \textbf{1980},
Springer, 2009.

\bibitem{BJ1981}
B.~Gaveau and J.~Vauthier,
{Intégrales oscillantes stochastiques : L'équation de Pauli}, \textit{J. Funct. Anal.}
\textbf{44} (1981) 388--400.

\bibitem{Hir2019}
F.~Hiroshima,
Non-rela\-tivistic limit of the semi-rela\-tivistic Pauli-Fierz model,
\textit{RIMS K\=o\-ky\=u\-roku} \textbf{2187} (2021) 1--10.

\bibitem{HIL2009}
F. Hiroshima, T. Ichinose, and J. L\H{o}rinczi, 
Path integral representation for Schr\"odinger operators with Bernstein functions of the Laplacian, 
\textit{Rev. Math. Phys.} \textbf{24} (2012) 1250013.

\bibitem{Hun1975}
W.~Hunziker,
{On the nonrelativistic limit of the Dirac theory},
\textit{Commun. Math. Phys.} \textbf{40} (1975) 215--222.

\bibitem{Ichi1987}
T.~Ichinose,
The nonrelativistic limit problem for a relativistic spinless particle in an electromagnetic field,
\textit{J. Funct. Anal.} \textbf{73} (1987) 233--257.

\bibitem{IT1986}
T.~Ichinose and H.~Tamura,
Imaginary-time path integral for a relativistic spinless particle in an electromagnetic field,
\textit{Commun. Math. Phys.} \textbf{105} (1986) 239--257.

\bibitem{HC1981}
H.~Leinfelder and C.~G.~Simader,
Schrödinger operators with singular magnetic
vector potentials,
\textit{Math. Z.} \textbf{176} (1981) 1--19.

\bibitem{JHV2020}
J.~L\H{o}rinczi, F.~Hiroshima, and V.~Betz,
\textit{Feynman-Kac type theorems and its applications 1}, 2nd ed,
De Gruyter, 2020.

\bibitem{RS1980}
M.~Reed and B.~Simon, 
\textit{Methods of Modern Mathematical Physics I: Functional Analysis},
revised and enlarged edition,
Academic Press, 1980.

\bibitem{Sim1979}
B.~Simon,
\textit{Functional Integration and Quantum Physics},
Academic Press, New York, 1979.

\bibitem{Tha1992}
B.~Thaller,
\textit{The Dirac Equation},
Springer-Verlag, Berlin Heidelberg, 1992.

\end{thebibliography}
\end{document}